\documentclass[aps,pra,twocolumn,superscriptaddress,showpacs,reprint]{revtex4-2}

\usepackage{amsthm, amsmath, amsfonts, amssymb}
\usepackage{braket}
\usepackage{graphicx}%
\usepackage{dcolumn}%
\usepackage{bm}%
\usepackage{hyperref}
\usepackage[capitalize]{cleveref}
\usepackage{quantikz}
\usepackage{makecell}
\usepackage{stmaryrd}

\usepackage{multirow}
\usepackage{comment}

\newtheorem{lemma}{Lemma}
\newtheorem{theorem}{Theorem}

\newtheorem{cor}{Corollary}
\usepackage[most]{tcolorbox}
\usepackage{longtable}
\usepackage{rotating}

\newcommand{\numphysicalqubitsnewfourerrorratefour}{4.76\times 10^{4}}
\newcommand{\targetDistancefourerrorratefour}{9}

\newcommand{\gateCountfourerrorratefour}{3.22\times 10^{6}}

\newcommand{\numphysicalqubitsnewfiveerrorratefour}{6.42\times 10^{4}}
\newcommand{\targetDistancefiveerrorratefour}{10}

\newcommand{\gateCountfiveerrorratefour}{3.85\times 10^{7}}

\newcommand{\numphysicalqubitsnewsixerrorratefour}{8.52\times 10^{4}}
\newcommand{\targetDistancesixerrorratefour}{11}

\newcommand{\gateCountsixerrorratefour}{4.48\times 10^{8}}

\newcommand{\numphysicalqubitsnewsevenerrorratefour}{1.15\times 10^{5}}
\newcommand{\targetDistancesevenerrorratefour}{12}

\newcommand{\gateCountsevenerrorratefour}{5.32\times 10^{9}}

\newcommand{\numphysicalqubitsneweighterrorratefour}{1.47\times 10^{5}}
\newcommand{\targetDistanceeighterrorratefour}{13}

\newcommand{\gateCounteighterrorratefour}{5.95\times 10^{10}}

\newcommand{\numphysicalqubitsnewnineerrorratefour}{1.85\times 10^{5}}
\newcommand{\targetDistancenineerrorratefour}{14}

\newcommand{\gateCountnineerrorratefour}{6.58\times 10^{11}}

\newcommand{\numphysicalqubitsnewtenerrorratefour}{2.36\times 10^{5}}
\newcommand{\targetDistancetenerrorratefour}{15}

\newcommand{\gateCounttenerrorratefour}{7.42\times 10^{12}}

\newcommand{\numphysicalqubitsnewelevenerrorratefour}{2.88\times 10^{5}}
\newcommand{\targetDistanceelevenerrorratefour}{16}

\newcommand{\gateCountelevenerrorratefour}{8.05\times 10^{13}}

\newcommand{\numphysicalqubitsnewtwelveerrorratefour}{3.48\times 10^{5}}
\newcommand{\targetDistancetwelveerrorratefour}{17}

\newcommand{\gateCounttwelveerrorratefour}{8.68\times 10^{14}}

\newcommand{\numphysicalqubitsnewthireteenerrorratefour}{4.25\times 10^{5}}
\newcommand{\targetDistancethireteenerrorratefour}{18}

\newcommand{\gateCountthireteenerrorratefour}{9.52\times 10^{15}}

\newcommand{\numphysicalqubitsnewfourteenerrorratefour}{5.03\times 10^{5}}
\newcommand{\targetDistancefourteenerrorratefour}{19}

\newcommand{\gateCountfourteenerrorratefour}{1.02\times 10^{17}}

\newcommand{\numphysicalqubitsnewfifteenerrorratefour}{5.89\times 10^{5}}
\newcommand{\targetDistancefifteenerrorratefour}{20}

\newcommand{\gateCountfifteenerrorratefour}{1.08\times 10^{18}}

\newcommand{\numLogicalQubitsfourerrorratethree}{217}

\newcommand{\numphysicalqubitsnewfourerrorratethree}{1.42\times 10^{5}}
\newcommand{\targetDistancefourerrorratethree}{17}

\newcommand{\classicalresourcefourerrorratethree}{2.09\times 10^{2}}
\newcommand{\classicalresourcebestfourerrorratethree}{1.20\times 10^{1}}

\newcommand{\numLogicalQubitsfiveerrorratethree}{259}

\newcommand{\numphysicalqubitsnewfiveerrorratethree}{2.03\times 10^{5}}
\newcommand{\targetDistancefiveerrorratethree}{19}

\newcommand{\classicalresourcefiveerrorratethree}{6.62\times 10^{2}}

\newcommand{\numLogicalQubitssixerrorratethree}{301}

\newcommand{\numphysicalqubitsnewsixerrorratethree}{2.82\times 10^{5}}
\newcommand{\targetDistancesixerrorratethree}{21}

\newcommand{\classicalresourcesixerrorratethree}{2.10\times 10^{3}}
\newcommand{\classicalresourcebestsixerrorratethree}{1.25\times 10^{2}}

\newcommand{\numLogicalQubitssevenerrorratethree}{357}

\newcommand{\numphysicalqubitsnewsevenerrorratethree}{3.94\times 10^{5}}
\newcommand{\targetDistancesevenerrorratethree}{23}

\newcommand{\classicalresourcesevenerrorratethree}{6.63\times 10^{3}}
\newcommand{\classicalresourcebestsevenerrorratethree}{3.95\times 10^{2}}

\newcommand{\numLogicalQubitseighterrorratethree}{399}

\newcommand{\numphysicalqubitsneweighterrorratethree}{5.56\times 10^{5}}
\newcommand{\targetDistanceeighterrorratethree}{26}

\newcommand{\classicalresourceeighterrorratethree}{2.10\times 10^{4}}
\newcommand{\classicalresourcebesteighterrorratethree}{1.25\times 10^{3}}

\newcommand{\numLogicalQubitsnineerrorratethree}{441}

\newcommand{\numphysicalqubitsnewnineerrorratethree}{7.08\times 10^{5}}
\newcommand{\targetDistancenineerrorratethree}{28}

\newcommand{\classicalresourcenineerrorratethree}{6.63\times 10^{4}}
\newcommand{\classicalresourcebestnineerrorratethree}{3.96\times 10^{3}}

\newcommand{\numLogicalQubitstenerrorratethree}{497}

\newcommand{\numphysicalqubitsnewtenerrorratethree}{9.11\times 10^{5}}
\newcommand{\targetDistancetenerrorratethree}{30}

\newcommand{\classicalresourcetenerrorratethree}{2.10\times 10^{5}}
\newcommand{\classicalresourcebesttenerrorratethree}{1.25\times 10^{4}}

\newcommand{\numLogicalQubitselevenerrorratethree}{539}

\newcommand{\numphysicalqubitsnewelevenerrorratethree}{1.12\times 10^{6}}
\newcommand{\targetDistanceelevenerrorratethree}{32}

\newcommand{\classicalresourceelevenerrorratethree}{6.63\times 10^{5}}
\newcommand{\classicalresourcebestelevenerrorratethree}{3.96\times 10^{4}}

\newcommand{\numLogicalQubitstwelveerrorratethree}{581}

\newcommand{\numphysicalqubitsnewtwelveerrorratethree}{1.36\times 10^{6}}
\newcommand{\targetDistancetwelveerrorratethree}{34}

\newcommand{\classicalresourcetwelveerrorratethree}{2.10\times 10^{6}}
\newcommand{\classicalresourcebesttwelveerrorratethree}{1.25\times 10^{5}}

\newcommand{\numLogicalQubitsthireteenerrorratethree}{637}

\newcommand{\numphysicalqubitsnewthireteenerrorratethree}{1.67\times 10^{6}}
\newcommand{\targetDistancethireteenerrorratethree}{36}

\newcommand{\classicalresourcethireteenerrorratethree}{6.63\times 10^{6}}
\newcommand{\classicalresourcebestthireteenerrorratethree}{3.96\times 10^{5}}

\newcommand{\numLogicalQubitsfourteenerrorratethree}{679}

\newcommand{\numphysicalqubitsnewfourteenerrorratethree}{1.98\times 10^{6}}
\newcommand{\targetDistancefourteenerrorratethree}{38}

\newcommand{\classicalresourcefourteenerrorratethree}{2.10\times 10^{7}}
\newcommand{\classicalresourcebestfourteenerrorratethree}{1.25\times 10^{6}}

\newcommand{\numLogicalQubitsfifteenerrorratethree}{721}

\newcommand{\numphysicalqubitsnewfifteenerrorratethree}{2.32\times 10^{6}}
\newcommand{\targetDistancefifteenerrorratethree}{40}

\newcommand{\classicalresourcefifteenerrorratethree}{6.63\times 10^{7}}
\newcommand{\classicalresourcebestfifteenerrorratethree}{3.96\times 10^{6}}

\newcommand{\numphysicalqubitsnewbbfourerrorratefour}{2.78\times 10^{4}}
\newcommand{\errorRatebbfourerrorratefour}{5.43\times 10^{-09}}

\newcommand{\numphysicalqubitsnewbbfiveerrorratefour}{3.09\times 10^{4}}
\newcommand{\errorRatebbfiveerrorratefour}{4.55\times 10^{-10}}

\newcommand{\numphysicalqubitsnewbbsixerrorratefour}{3.39\times 10^{4}}
\newcommand{\errorRatebbsixerrorratefour}{3.91\times 10^{-11}}

\newcommand{\numphysicalqubitsnewbbsevenerrorratefour}{3.78\times 10^{4}}
\newcommand{\errorRatebbsevenerrorratefour}{3.29\times 10^{-12}}

\newcommand{\numphysicalqubitsnewbbeighterrorratefour}{4.08\times 10^{4}}
\newcommand{\errorRatebbeighterrorratefour}{2.94\times 10^{-13}}

\newcommand{\numphysicalqubitsnewbbnineerrorratefour}{4.39\times 10^{4}}
\newcommand{\errorRatebbnineerrorratefour}{2.66\times 10^{-14}}

\newcommand{\numphysicalqubitsnewbbtenerrorratefour}{4.78\times 10^{4}}
\newcommand{\errorRatebbtenerrorratefour}{2.36\times 10^{-15}}

\newcommand{\numphysicalqubitsnewbbelevenerrorratefour}{5.01\times 10^{4}}
\newcommand{\errorRatebbelevenerrorratefour}{2.17\times 10^{-16}}

\newcommand{\numphysicalqubitsnewbbtwelveerrorratefour}{5.31\times 10^{4}}
\newcommand{\errorRatebbtwelveerrorratefour}{2.02\times 10^{-17}}

\newcommand{\numphysicalqubitsnewbbthireteenerrorratefour}{5.70\times 10^{4}}
\newcommand{\errorRatebbthireteenerrorratefour}{1.84\times 10^{-18}}

\newcommand{\numphysicalqubitsnewbbfourteenerrorratefour}{7.97\times 10^{4}}
\newcommand{\errorRatebbfourteenerrorratefour}{1.72\times 10^{-19}}

\newcommand{\numphysicalqubitsnewbbfifteenerrorratefour}{8.41\times 10^{4}}
\newcommand{\errorRatebbfifteenerrorratefour}{1.62\times 10^{-20}}

\begin{document}

\title{Realization of a Quantum Streaming Algorithm on Long-lived Trapped-ion Qubits}%

\author{Pradeep Niroula}
\email{pradeep.niroula@jpmchase.com}
\affiliation{Global Technology Applied Research, JPMorganChase, New York, NY 10017, USA}

\author{Shouvanik Chakrabarti}
\affiliation{Global Technology Applied Research, JPMorganChase, New York, NY 10017, USA}
\author{Steven Kordonowy}
\affiliation{Global Technology Applied Research, JPMorganChase, New York, NY 10017, USA}
\author{Niraj Kumar}
\affiliation{Global Technology Applied Research, JPMorganChase, New York, NY 10017, USA}
\author{Sivaprasad Omanakuttan}
\affiliation{Global Technology Applied Research, JPMorganChase, New York, NY 10017, USA}
\author{Michael A. Perlin}
\affiliation{Global Technology Applied Research, JPMorganChase, New York, NY 10017, USA}

\author{M.S. Allman}
\affiliation{Quantinuum, Broomfield, CO 80021, USA}
\author{J.P. Campora III}
\affiliation{Quantinuum, Broomfield, CO 80021, USA}
\author{Alex~Chernoguzov}
\affiliation{Quantinuum, Broomfield, CO 80021, USA}
\author{Samuel F. Cooper}
\affiliation{Quantinuum, Broomfield, CO 80021, USA}
\author{Robert D. Delaney}
\affiliation{Quantinuum, Broomfield, CO 80021, USA}
\author{Joan M. Dreiling}
\affiliation{Quantinuum, Broomfield, CO 80021, USA}
\author{Brian Estey}
\affiliation{Quantinuum, Broomfield, CO 80021, USA}
\author{Caroline Figgatt}
\affiliation{Quantinuum, Broomfield, CO 80021, USA}
\author{Cameron Foltz}
\affiliation{Quantinuum, Broomfield, CO 80021, USA}
\author{John P. Gaebler}
\affiliation{Quantinuum, Broomfield, CO 80021, USA}
\author{Alex Hall}
\affiliation{Quantinuum, Broomfield, CO 80021, USA}
\author{Ali A. Husain}
\affiliation{Quantinuum, Brooklyn Park, MN 55422, USA}
\author{Akhil Isanaka}
\affiliation{Quantinuum, Broomfield, CO 80021, USA}
\author{Colin J. Kennedy}
\affiliation{Quantinuum, Broomfield, CO 80021, USA}
\author{Nikhil Kotibhaskar}
\affiliation{Quantinuum Partnership House, London SW1P 1BX, UK}
\author{Ivaylo S. Madjarov}
\affiliation{Quantinuum, Broomfield, CO 80021, USA}
\author{Michael Mills}
\affiliation{Quantinuum, Broomfield, CO 80021, USA}
\author{Alistair R. Milne}
\affiliation{Quantinuum Partnership House, London SW1P 1BX, UK}
\author{Louis Narmour}
\affiliation{Quantinuum, Broomfield, CO 80021, USA}
\author{Annie J. Park}
\affiliation{Quantinuum, Broomfield, CO 80021, USA}
\author{Adam P. Reed}
\affiliation{Quantinuum, Broomfield, CO 80021, USA}
\author{Kartik Singhal}
\affiliation{Quantinuum, Broomfield, CO 80021, USA}
\author{Anthony Ransford}
\affiliation{Quantinuum, Broomfield, CO 80021, USA}
\author{Justin G. Bohnet}
\affiliation{Quantinuum, Broomfield, CO 80021, USA}
\author{Brian Neyenhuis}
\affiliation{Quantinuum, Broomfield, CO 80021, USA}
\author{Rob~Otter}
\affiliation{Global Technology Applied Research, JPMorganChase, New York, NY 10017, USA}
\author{Ruslan~Shaydulin}
\email{ruslan.shaydulin@jpmchase.com}
\affiliation{Global Technology Applied Research, JPMorganChase, New York, NY 10017, USA}

\date{\today}%

\begin{abstract}
Large classical datasets are often processed in the streaming model, with data arriving one item at a time. In this model, quantum algorithms have been shown to offer an unconditional exponential advantage in space. However, experimentally implementing such streaming algorithms requires qubits that remain coherent while interacting with an external data stream. In this work, we realize such a data-streaming model using Quantinuum Helios trapped-ion quantum computer with long-lived qubits that communicate with an external server. We implement a quantum pair sketch, which is the primitive underlying many quantum streaming algorithms, and use it to solve Hidden Matching, a problem known to exhibit a theoretical exponential quantum advantage in space.  
Furthermore, we compile the quantum streaming algorithm to fault-tolerant quantum architectures based on surface and bivariate bicycle codes and show that the quantum space advantage persists even with the overheads of fault-tolerance.
\end{abstract}

\maketitle

\section*{Introduction}
In many practical settings, data is generated faster than it can be stored, or the total volume of data is so large that it cannot be loaded into local computer memory in its entirety.
These settings motivate the streaming model of computation, in which the data is
received and processed  %
one item at a time, %
while the algorithm maintains only a small sketch without ever storing the entire dataset. 
A sketch is a compact representation of the dataset that enables the algorithm to perform a restricted set of queries. 
The goal is to minimize the space requirements of the sketch required to achieve a given solution quality.
Algorithms for a wide range of problems have been developed in the streaming model~\cite{Muthukrishnan2005}. In this work, we restrict our attention to the single-pass model, wherein each data point is given to the algorithm only once.

Quantum algorithms are known to provide a space advantage in the streaming model~\cite{kallaugher2024howto, nadya, kallaugher_2022_quantum_triangle_counting, BHM1, BHM2, BHM3, 10.5555/3179430.3179435, 10.4230/lipics.ccc.2017.23}. For natural graph problems, quantum algorithms have been proposed that use polynomially~\cite{kallaugher_2022_quantum_triangle_counting} or even exponentially~\cite{nadya} fewer qubits than the number of bits required by the best possible classical algorithm. Importantly, unlike for many instances of superpolynomial quantum advantage in runtime, such as in factoring large integers or sampling from random circuits, the space advantages are unconditional as it is possible to derive a lower bound on the number of bits required by any classical algorithm to achieve a given solution quality. Such lower bounds are typically not available for classical runtimes. Interestingly, the exponential space advantage in the streaming model contrasts with the at-most-quadratic advantage possible in the standard data access model~\cite{watrous1999space}.

Experimental realization of quantum streaming algorithms requires the adaptive application of quantum gates conditioned on a classical data stream. Although related primitives have been demonstrated in the context of error correction \cite{google2025quantum,Daguerre2025} and adaptive quantum circuits~\cite{Crcoles2021,2302.03029,Bumer2024,CarreraVazquez2024}, the streaming model introduces an additional challenge of network latency, requiring the quantum state to remain coherent throughout the data stream. The need for long coherence times poses major challenges for some hardware platforms, as qubit coherence times (e.g., less than a millisecond for superconducting transmon qubits~\cite{google2025quantum}) must exceed the network latency (typically on the order of milliseconds on the cloud).

The quantum pair sketch \cite{kallaugher2024howto} was recently shown to unify all known quantum streaming algorithms for graph problems that achieve a space advantage in the single-pass model.
Specifically, quantum algorithms for Hidden Matching, triangle counting, and maximum directed cut problems can all be viewed as hybrid quantum-classical algorithms that use the quantum processor solely to manipulate the quantum pair sketch. 
Consequently, implementing the quantum pair sketch opens the door to realizing all these streaming algorithms.

\begin{figure*}[t]
    \centering
 \includegraphics[width=\linewidth]{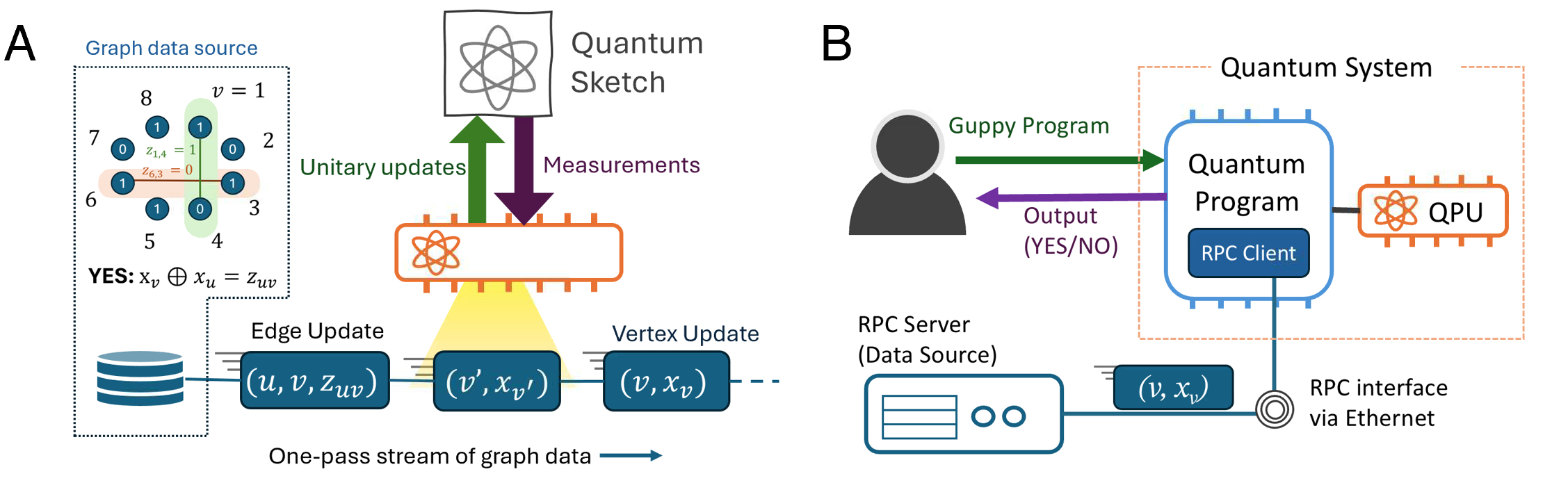}
    \caption{\textbf{A}) \textbf{Overview of the quantum streaming algorithm for Hidden Matching.} A quantum computer with sequential, one-time access to a data stream maintains a quantum sketch on which it can apply unitary gates and perform measurements as it reads the data stream. In the Hidden Matching problem, the data stream consists of vertex $(v, x_v)$ and edge-updates $(u, v, z_{uv})$, where $u,v$ are vertices, and $x_{v}$ and $z_{uv}$ are vertex and edge labels, respectively. In the YES instance of the problem (illustrated), the relation $x_{u}\oplus x_{v} = z_{uv}$ holds for every edge. \textbf{B) Experimental Setup.}  A user submits a program, written in the high-level language Guppy \cite{koch2025guppy}, %
    to a system that consists of a quantum processor (QPU) and a classical control system. The control system contains a remote procedure call (RPC) client that queries a streaming server initialized with a Hidden Matching problem instance for vertex and edge updates (only vertex updates are illustrated), and generates quantum instructions, based on the received updates, in real-time. Upon termination, the quantum system returns the output to the user.
}
    \label{fig:streaming-overview}
\end{figure*}

In this work, we realize the quantum pair sketch and use it to execute a quantum streaming algorithm for Hidden Matching problem on a trapped-ion quantum processor. 
Our experiment %
is enabled by a series of technical advancements that allow the quantum processor to interact with an externally-hosted data-stream. Although such a remote communication capability has broad applications -- for example, in reducing latencies in certified randomness protocols based on random circuit sampling~\cite{jpmccertrand2} -- we use this capability towards a problem with unconditional advantage in space resources. 

The fidelity of the quantum processor in our experiment is insufficient to achieve an unconditional space advantage, necessitating the introduction of fault tolerance. 
We compile both the quantum sketch and the end-to-end Hidden Matching algorithm to fault-tolerant architectures based on the surface and bivariate bicycle codes, showing that only tens of thousands of physical qubits are required to obtain a space advantage.
Although further algorithmic improvements are needed before quantum advantage becomes practical, our results demonstrate that the streaming model offers a viable route to experimental resource advantage for natural problems.

\section*{Exponential Quantum Space Advantage for Hidden Matching Problem}

We use the streaming capability to implement the quantum pair sketch to solve the streaming version of the Hidden Matching (HM) problem, illustrated in Fig.~\ref{fig:streaming-overview}A. An HM problem is parameterized by the size (number of nodes) of a graph, $n$, and a parameter $\alpha \in [0, 1/4]$. An HM problem instance is specified by three inputs: (1) a bitstring $x \in \{ 0, 1\}^n$ , (2) an $\alpha n$-edge (partial) matching $M$ (i.e., a set of $\lvert M\rvert=\alpha n$ pairwise non-adjacent edges), and (3) an additional bitstring $z \in \{0, 1\}^{|M|}$. The bitstring $x$ represents a binary assignment of the set of vertices $[n]$, with $x_v$ denoting the assignment of vertex $v$. $M$ is a (partial) matching of a graph on $V = [n]$, and $z$ contains the edge labels of $M$. The HM problem promises that either the string $x$ matches the edge labels $z$ for all edges (YES case), or $x$ does not match $z$ for any of the edges (NO case). That is, in the YES case $x_u \oplus x_v = z_{uv}$ for all $(u,v) \in M$, while in the NO case $x_v \oplus x_v =  \overline{z_{uv}}$ for all $(u,v) \in M$ .

The stream is composed of vertex updates $(v, x_v)$ and edge updates $(u, v, z_{uv})$. 
The goal of the streaming algorithm is to decide the promise problem (distinguish between the YES and NO cases) better than random guess while preserving as little of the data-stream as possible. Although HM is the simplest setting for illustrating quantum space advantage, the HM protocol may be adapted into schemes for quantum money \cite{kumar2019practically} or as a strong extractor of randomness \cite{BHM3}. Furthermore, the same quantum pair sketch also underpins algorithms for triangle counting and maximum directed cut~\cite{kallaugher2024howto}. Details of the compilation of the general quantum pair sketch are provided in Appendix~\ref{sec:quantum_pair_sketch_complexity}.

The asymptotically optimal classical algorithm for the HM problem, detailed in \cref{sec:classical_supp}, stores a randomly sampled subset of the incoming $n$ vertices and checks if an edge $(u,v)$ in the partial matching belongs to the sampled subset. This mirrors the scaling necessary for a collision of independent events (for example, the birthday problem) and has a space complexity of $O(\sqrt{n/\alpha})$. In particular, for a success probability of $2/3$, the best known classical algorithm requires at least $\lceil  \sqrt{\ln(3) n/\alpha} \rceil$ bits of information (see \cref{sec:classical_supp} for a derivation of the constant). 

While the classical algorithm outlined above is constructive, communication complexity arguments yield a tighter asymptotic lower bound of $\Omega(\sqrt{n})$ for all possible classical algorithms \cite{bar2004exponential}. In order to facilitate a direct comparison with the quantum algorithm, we determine the exact lower-bound to be $\sqrt{(n-1)/\alpha}/(6e\sqrt{2}\ln(2))$. The details of this calculation are presented in \cref{sec:classical_supp_lower_bound}.  

On the other hand, a quantum sketch for HM may store useful information about $n$ vertices using only $O(\log(n))$ qubits \cite{BHM2, kallaugher2024howto}. In particular, for a success probability of $2/3$, the algorithm only uses at most $\lceil 1.5 \cdot \left(\log(n) + 2\right) / \alpha \rceil$ qubits (See \cref{sec:quantum-supp} for how we obtain the constants). In the hybrid algorithm for HM, the quantum stage is followed by a classical stage that stores only a single edge. We describe the algorithm below and refer the reader to Ref.~\cite{kallaugher2024howto} for a pedagogical introduction.

The quantum algorithm initializes the sketch in the superposition $\frac{1}{\sqrt{2n}} \sum_{v \in [n], b \in \{0,1\}}\ket{v,0,b}$. The three registers respectively encode the vertex index $v$, its label $x_v$, and an ancillary bit used for parity measurements. In the beginning of the algorithm, all of the vertices are labeled $x_v=0$, and this label is progressively updated as vertex information is streamed to the program. Storing such a quantum workspace requires $\lceil \log{n} \rceil + 2$ qubits.

A classical control system  mediates the interaction between the external network and the quantum state.
Upon reading a vertex update $(v, x_v)$ from the stream, if $x_v = 1$, the controller it sends instructions to the quantum computer to update $\ket{v,0,b} \mapsto \ket{v,1,b}$ for both $b \in \{0,1\}$. This update is implemented with a multi-controlled NOT gate that flips the middle register conditioned on the first register.  For simplicity of exposition, we assume that vertices arrive before edges, although the algorithm does not rely on this ordering.
This choice of ordering does not impact the lower bound for the classical algorithm.
We note that with this ordering, the classical stage of the hybrid algorithm specified in Ref.~\cite{kallaugher2024howto} is never invoked. 

Upon receiving an edge update $(u, v, z_{uv})$, the quantum computer makes a sequence of four measurements described by a projective-valued measure (PVM)
defined as follows.
For $a,b \in \{0,1\}$, let $\Pi_{u,v;a,b}^\pm$ be the projection onto $\ket{u,a,a\oplus b} \pm \ket{v,b,a\oplus b}$, and let $\Pi_{u,v;a,b}^0 = I - \Pi_{u,v;a,b}^+ - \Pi_{u,v;a,b}^-$ be the projector orthogonal to both $\Pi_{u,v;a,b}^+$ and $\Pi_{u,v;a,b}^-$. The PVM used by the quantum algorithm is
\begin{align}
    \mathcal{M}_{u,v;a,b} := \left\{ \Pi_{u,v;a,b}^+, \Pi_{u,v;a,b}^-,
    \Pi_{u,v;a,b}^0
    \right\}.
\end{align}
We sequentially perform the PVM measurements for all combination of $a, b \in \{0, 1\}$, resulting in up to four measurements per edge. If any measurement yields $\Pi_{u,v;a,b}^+$ or $\Pi_{u,v;a,b}^-$, the program terminates; otherwise, it proceeds onto the measurement or eventually the next edge. In particular, measurement of $\Pi_{u,v;a,b}^+$ suggests the observation of a pair with $x_u = a$ and $x_v = b$, and the program thus outputs YES if $a \oplus b \oplus z_{uv}$ is zero, and NO otherwise. If $\Pi_{u,v;a,b}^-$ is observed, the program terminates, outputting NULL. Similarly, if the stream reaches the end, without ever measuring $\Pi_{u,v;a,b}^{\pm}$ the program terminates with NULL. 
These steps are illustrated in Fig.~\ref{fig:streaming-overview}A. 
This also means that the length of the program is dynamic. For a problem with $\alpha n$ edges in the stream, the program may terminate after any of the $4\alpha n$ PVM measurements.

Under idealized noiseless quantum execution, the probability of a quantum sketch returning a correct answer $p_{\rm correct}$ is $\alpha$, the probability of it returning an incorrect answer $p_{\rm wrong}$ is at most $\alpha/2$ and the probability of it returning the null answer is $p_{\rm{null}} = 1-p_{\rm correct} - p_{\rm wrong}$.
The probability of solving an HM problem can be boosted to 2/3
by running this algorithm in parallel on $O(1/\alpha)$-many copies and taking the majority vote between YES and the NO results.
The number of necessary copies may be determined by numerically calculating overall success probabilities after majority vote; see \cref{fig:BHM_majority_vote} in \cref{sec:quantum-supp}. Furthermore, these probabilities are affected by noise in our experiments. 
A simple noise model illustrates this effect: if the final state is depolarized, $\tilde{\rho} = \gamma \rho + (1-\gamma) I$, the gap between success and failure probabilities is roughly scaled to $\gamma \alpha/2$, so the number of copies required to boost the overall success of majority vote grows as $O(1/\gamma \alpha)$.

\begin{figure*}[!ht]
    \centering
    \includegraphics[width=\linewidth]{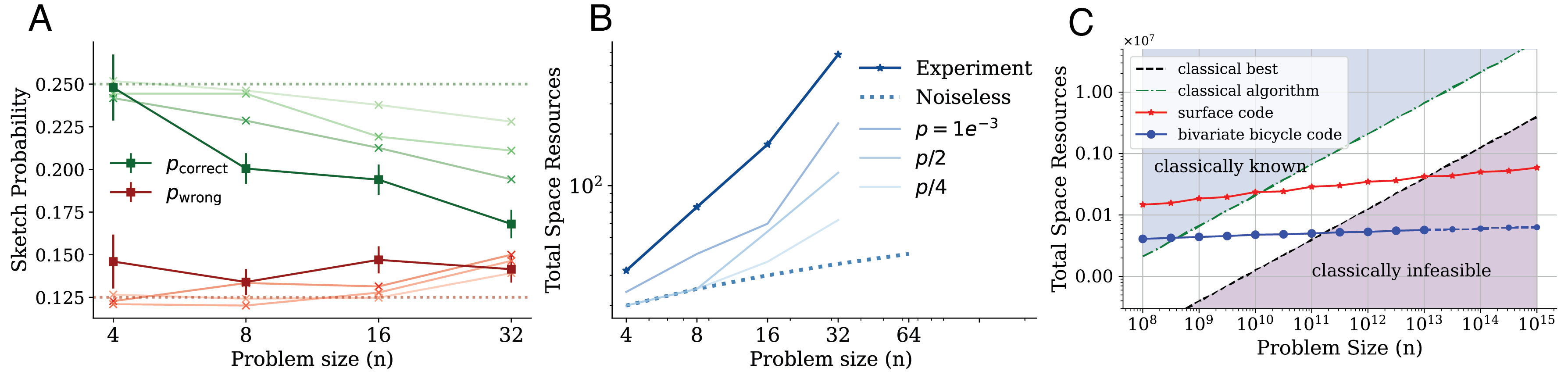}
    \caption{\textbf{A}) \textbf{Experimental results.} Sketch fidelities---the probabilities that the sketch yields the correct result (green) or an incorrect result (red)---for different problem sizes and simulated noise levels (lighter shades; two-qubit gate depolarization errors are given in the legend for panel B). The experimental probabilities (dark solid lines) were obtained from experiments with $500, 2000, 2000$ and $2000$ shots for $n=4,8, 16$ and $32$ respectively. The ideal sketch fidelity (horizontal dotted lines) is independent of the problem size, giving $p_{\rm correct} = 1/4$ and $p_{\rm wrong} = 1/8$.  
    \textbf{B}) \textbf{Space overhead for boosting success probability.} Given the sketch fidelities at different problem sizes and noise strengths, the total quantum space required to obtain an algorithm success probability of $2/3$ using majority vote, for noisy experiments (solid blue line) and simulations with two-qubit depolarizing rates (see legend for error rates $p$). Dotted blue line shows the total space required for a noiseless implementation. Classical space requirements are lower than noiseless quantum and are not pictured. 
    \textbf{C}) \textbf{Fault-tolerant break-even conditions.} The total space required to break even with the best known classical algorithm (green) and a classical lower bound (black) is calculated using both the surface code (red) and a bivariate bicycle code (solid blue line for the two-gross code, dashed blue line for a higher-distance bivariate bicycle code; see \cref{sec:FT_resource_estimation} for details). 
    The quantum resource counts include the physical qubits required for $T$ and CCZ resource-state factories.
    }  
    \label{fig:experiment}
\end{figure*}

\section*{Experimental Realization of~Quantum~Streaming~Model}

To enable the quantum streaming model, a quantum computer must maintain coherence in its qubits while requesting information about how to proceed with the execution of a quantum program. Low latency and minimal qubit idle times are critical for sustaining coherent evolution during this process.
To meet these requirements, we developed three software components.
First, we implemented a real-time execution engine that runs compiled LLVM programs natively on Quantinuum’s Helios control system, which contains a \textit{Quantum Instruction Set} (QIS) to execute native gates, initialization, and measurement on ``virtual qubits’’~\cite{cross2022openqasm}. Second, we developed Guppy~\cite{koch2025guppy}, a high-level programming language that compiles to this QIS executable format.  Finally, we introduced the Helios runtime \textit{Helios runtime}~\cite{helios}, which in real-time continuously receives gating operations on virtual qubits from the executing QIS program and just-in-time compiles these gates into operations on Helios's physical qubits. 
The resulting physical instructions comprise transport primitives that transport ions into quantum logic regions, and ``quantum logic’’ operations that implement parallelized 1Q and 2Q native gates, initialization, and measurement~\cite{helios}.

In addition to the software above, this experiment is enabled by a real-time remote procedure call (RPC) interface that enables communication over a high-speed network with the software outlined above, allowing classical data to be exchanged, processed, and reintegrated into the ongoing quantum evolution being driven by the executing quantum program.
The experimental setup consists of a classical streaming server (connected to the control system via Ethernet, Fig.~\ref{fig:streaming-overview}B),  the three pieces of software discussed in the previous paragraph that execute together on the control system, and the RPC interface between the two aforementioned components.

The streaming server listens for RPC requests from the executing program and returns information to the control system that determines how the evolution of quantum state will proceed. The executing QIS program is compiled with these RPCs and contains compiled code to unpack the data received from the streaming server into an array that then can be interpreted using high-level programming features in Guppy, such as loops that index into this array and branching conditioned on the data in this array that specifies what gates to apply on allocated virtual qubits. Once the Helios runtime receives the gate list for virtual qubits (generated based on the data stream), it locates the corresponding physical qubits and runs a sorting algorithm that transports them into the quantum logic regions where the gates are executed. It does so with minimal latency,
and the use of real-time location tracking prevents unnecessary ion transport for conditional gates that ultimately are not executed.

\medskip

\begin{table}[!ht]
\begin{tabular}{|c|c|cccc|}
\hline
\multirow{2}{*}{$n$} & \multirow{2}{*}{Space} & \multicolumn{4}{c|}{Operations}                                                                                            \\ \cline{3-6} 
                     &                        & \multicolumn{1}{c|}{$H$}   & \multicolumn{1}{c|}{$CX$}              & \multicolumn{1}{c|}{$C^{\log(n)+1}X$} & $C^{\log(n)+3}X$ \\ \hline
4                    & 4                      & \multicolumn{1}{c|}{10}   & \multicolumn{1}{c|}{28}              & \multicolumn{1}{c|}{4}                & 8                \\ \hline
8                    & 5                      & \multicolumn{1}{c|}{20}   & \multicolumn{1}{c|}{75}              & \multicolumn{1}{c|}{8}                & 16               \\ \hline
16                   & 6                     & \multicolumn{1}{c|}{36}  & \multicolumn{1}{c|}{186}             & \multicolumn{1}{c|}{16}               & 32               \\ \hline
32                   & 7                     & \multicolumn{1}{c|}{69}  & \multicolumn{1}{c|}{441}             & \multicolumn{1}{c|}{32}               & 64               \\ \hline
64                   & 8                     & \multicolumn{1}{c|}{134}  & \multicolumn{1}{c|}{1016}            & \multicolumn{1}{c|}{64}               & 128              \\ \hline
$n$                  & \makecell{$\log(n)$\\ $+ 2$}         & \multicolumn{1}{c|}{$2n+\log n$} & \multicolumn{1}{c|}{\makecell{$(2n-1)\cdot$ \\ $(\log n + 2)$}} & \multicolumn{1}{c|}{$n$}              & $2n$             \\ \hline
\end{tabular}
\caption{\textbf{Logical resource estimates for the HM algorithm.} Space and worst-case logical gate counts to run the hybrid quantum-classical streaming algorithm for $HM$ with $\alpha=1/4$.  When implementing this algorithm non-fault tolerantly, each $k$-controlled $X$ (i.e., $C^k X$) gate gets decomposed into $k-1$ three-qubit Toffoli ($CCX$) gates. 
}
\label{tab:gate_count}
\end{table}

{
\renewcommand{\arraystretch}{1.5}
\begin{table*}[!ht]
\centering
\begin{tabular}{|c|c|c|c|c|c|c|c|c|}
\hline
\multirow{3}{*}{$n$} 
  & \multirow{3}{*}{\makecell*[c]{logical\\qubits}} 
  & \multirow{3}{*}{\makecell*[c]{Toffoli\\gates}} 
  & \multicolumn{2}{c|}{surface code} 
  & \multicolumn{1}{c|}{two-gross code}
  & \multicolumn{2}{c|}{classical bits} \\
\cline{4-8}
  & & & \multicolumn{2}{c|}{physical qubits}  
    & \multicolumn{1}{c|}{physical qubits  }
    & \multirow{2}{*}[0.5ex]{\makecell*[c]{best \\ known}}  
    & \multirow{2}{*}[0.5ex]{\makecell*[c]{lower\\ bound}} \\ 
\cline{4-6}
  & & & $p=10^{-3}$ & $p=10^{-4}$  & $p=10^{-4}$ & & \\
\hline 
$10^{10}$  & \numLogicalQubitstenerrorratethree & $\gateCounttenerrorratefour$ 
  & $\numphysicalqubitsnewtenerrorratethree$ 
  & $\numphysicalqubitsnewtenerrorratefour$
  & $\numphysicalqubitsnewbbtenerrorratefour$ 
  & $\classicalresourcetenerrorratethree$
  & $\classicalresourcebesttenerrorratethree$ \\
\hline
$10^{11}$  & \numLogicalQubitselevenerrorratethree & $\gateCountelevenerrorratefour$ 
  & $\numphysicalqubitsnewelevenerrorratethree$ 
  & $\numphysicalqubitsnewelevenerrorratefour$
  & $\numphysicalqubitsnewbbelevenerrorratefour$ 
  & $\classicalresourceelevenerrorratethree$
  & $\classicalresourcebestelevenerrorratethree$ \\
\hline
$10^{12}$  & \numLogicalQubitstwelveerrorratethree & $\gateCounttwelveerrorratefour$ 
  & $\numphysicalqubitsnewtwelveerrorratethree$ 
  & $\numphysicalqubitsnewtwelveerrorratefour$
  & $\numphysicalqubitsnewbbtwelveerrorratefour$ 
  & $\classicalresourcetwelveerrorratethree$
  & $\classicalresourcebesttwelveerrorratethree$ \\
\hline
\end{tabular}
\caption{
\textbf{Fault-tolerant resource estimates for the HM algorithm.}
Resource estimates are shown for the HM algorithm, assuming a single CCZ factory for Toffoli gates. 
We consider fault-tolerant implementations based on the rotated surface code with physical error rates of $p=10^{-3}$ and $p=10^{-4}$, and on the two-gross code---a bivariate bicycle code with known overheads for fault-tolerant quantum computation~\cite{yoder2025tour}---with a physical error rate of $p=10^{-4}$.
 See \cref{sec:FT_resource_estimation} for additional details.}
\label{tab:Resource_estimation_BHM_algorithm}
\end{table*}
}

We use the described architecture to implement a classical data source and a quantum program implementing the HM for graphs of size $n=4, 8, 16$, and $32$.
During initialization, the streaming source (external host) samples an instance of the HM problem. The user uploads a quantum program that repeatedly queries the external host for vertex and edge updates and then terminates with a YES, NO or NULL response. By repeating the procedure many times, we obtain statistics for the probabilities that a quantum sketch returns a correct answer, an incorrect answer, or a null outcome. 

Quantum pair sketch relies on four elementary operations, listed in \cref{sec:quantum_pair_sketch_complexity}. We compile these operations specifically for HM (\cref{sec:hm-quantum-gate-counts}) for our experiment. In our non-fault-tolerant implementation, the longest possible circuits (those that process all vertices and edges without terminating with a YES or a NO answer) consist of 195, 555, 1434, and 3513 CNOT gates for problem sizes $n=4, 8, 16,$ and $32$, respectively (See \cref{tab2:CX_counts} and \cref{sec:hm-quantum-gate-counts} for the scaling of single- and two-qubit physical gates). In Figure.~\ref{fig:experiment}A, we plot the sketch success and failure probabilities for graphs of different sizes, derived from experiments with 500, 2000, 2000, and 2000 shots for graph sizes $n=4, 8, 16$, and $32$, respectively. 
The decay of the success probability with problem size is consistent with the expected reduction in fidelity.

The condition for a quantum space advantage is that the total number of physical qubits used by the quantum algorithm is lower than the lower bound on the number of bits required by any classical algorithm. In defining space advantage solely in terms of the number of qubits, we ignore the space requirements for the classical control system used to manipulate the qubits and classical cost of decoding, among others. However, we similarly do not count the space overhead of the operating system, programming environment and other systems required to run the classical algorithm. Consequently, we consider the break-even point solely in terms of physical qubits and classical bits used by respective algorithms, excluding other overheads associated with running both quantum and classical algorithms in practice.

To compute the total space requirements, we use the sketch outcome probabilities ($p_{\rm correct}, p_{\rm wrong}$, and $p_{\rm null}$) to infer the number of sketches required to amplify the difference between $p_{\rm correct}$ and $p_{\rm wrong}$ to achieve a target algorithm success probability. In Figure.~\ref{fig:experiment}B, we plot the total number of qubits required to achieve a success probability of $2/3$. In addition, we use the quantum emulator to simulate the Quantinuum Helios device at various noise levels. Plotting in Figure.~\ref{fig:experiment}B the total space (size of sketch times number of sketches) needed at different relative noise levels, we see that at high error levels, the total space necessary increases without ever breaking even with the classical space. Indeed, even with a very generous error model where sketch fidelity only decays as $\gamma = O(1/n)$, the number of copies required to boost the algorithm success scales as $O(n)$, thereby eliminating any advantage of using a logarithmic-sized sketch. As a result, a successful realization of space advantage in the HM problem requires fault-tolerance.

\section*{Prospects for a Practical Fault-tolerant Quantum Space Advantage}

Estimating the break-even point between the classical and quantum approaches requires calculating the physical overhead of executing the algorithm fault-tolerantly while maintaining an acceptable logical error rate for each problem size. Noisy quantum sketches with fidelity $\gamma$ raise the failure probability of the algorithm from $\delta(k, \alpha)$ to $\delta'(k, \alpha, \gamma) = \delta(k, \alpha) + k (1-\gamma)$, where $k$ is the number of sketches used by the algorithm and $\alpha$ is the parameter of the HM problem. We numerically determine maximum tolerable infidelity, $1-\gamma$, that keeps the failure probability below $1/3$ (See \cref{fig:infidelity} in \cref{sec:FT_resource_estimation} for illustration at $\alpha=1/4$). We then count the number of expected logical qubits and logical quantum operations necessary to run each quantum sketch for varying problem sizes in \cref{sec:FT_resource_estimation}, and summarize the corresponding quantum resources in \cref{tab:gate_count}.

These logical gate counts yield the fidelity $F_*$ required for each logical quantum operation in order to execute the algorithm to an acceptable logical fidelity, which we set to $0.9975$, such that the overall failure probability of the noisy algorithm $\delta'(k=7, \alpha=1/4, \gamma=0.9975)$ is less than $1/3$.
The required fidelity $F_*$  in turn fixes the code distance---and hence the size and overhead---of the error-correcting code used for fault-tolerant execution.
To estimate the required \emph{physical} resources, we analyze fault-tolerant implementations based on the rotated surface code and the two-gross code, each equipped with a single CCZ magic-state factory to provide a universal gate set~\cite{gidney2025factor,Litinski2019magicstate,yoder2025tour}.
\cref{tab:Resource_estimation_BHM_algorithm} summarizes the physical resource requirements for these architectures at physical error rates of $10^{-3}$ and $10^{-4}$.

These resource estimates enable us to identify the break-even point between classical and quantum sketches.
In \cref{sec:classical_supp} and \cref{sec:quantum-supp}, we derive the number of classical bits needed to solve HM with success probability $2/$3. 
\cref{fig:experiment}C compares the classical space requirement with the physical qubit count  (for physical error rate $10^{-4}$) and shows that the quantum approach breaks even with the best classical method at $n \sim 10^{12}$, corresponding to $581$ logical qubits and $5\times10^4$ physical qubits. 
This indicates that an unconditional quantum space advantage in streaming model can be achieved with space requirements about an order of magnitude lower than those required for running Shor's factoring algorithm~\cite{gidney2025factor}.

We remark that the break-even point for the Hidden Matching problem corresponds to deep quantum circuits ($> 10^{12}$ Toffoli gates). 
For other applications of the quantum pair sketch like triangle counting and maximum directed cut, the break-even requires even greater resources. Consequently, making the quantum space advantage practical requires additional improvements to the quantum algorithms. In this work, we take the first steps towards this goal by fully compiling the quantum sketch and experimentally realizing the quantum streaming algorithm for the Hidden Matching problem.

\paragraph*{Note:} A recent work has leveraged arguments from communication complexity to demonstrate an unconditional separation between classical and quantum resources on the task of generating samples from a random circuit \cite{kretschmer2025demonstrating}. However, Ref.~\cite{kretschmer2025demonstrating} implements a narrow task that was designed specifically to realize quantum resource advantage and has no practical application. In contrast, the building blocks we demonstrate are directly applicable to natural graph problems like triangle counting~\cite{kallaugher_2022_quantum_triangle_counting} and maximum directed cut~\cite{nadya}. Furthermore, Ref.~\cite{kretschmer2025demonstrating} is analyzed in terms of one-way communication complexity, which is related to but distinct from the one-pass streaming model we consider.

\section*{Acknowledgments}
We thank Ojas Parekh and Nadezhda Voronova for helpful discussions on quantum streaming algorithms. We are grateful to Jamie Dimon, Lori Beer, and Scot Baldry for their executive support of JPMorganChase’s
Global Technology Applied Research Center and our work in quantum computing. We thank the technical staff at JPMorganChase’s Global Technology Applied Research for their support. 

\section*{Author Contributions}

P.N. and R.S. designed the experiment and led the overall project. P.N. and S.K. designed the circuits for the quantum pair sketch and the Hidden Matching algorithm, and analyzed their resource requirements. 
P.N. wrote the code implementing the quantum algorithms, ran the experiments, and analyzed the results.
P.N., S.K., S.O., and M.A.P. compiled the quantum algorithms to fault-tolerant architectures and performed resource estimation. 
P.N., S.K., S.C., and N.K. derived the lower bounds for classical algorithms. 
A.C. developed the real-time RPC concept. K.S., A.C. developed the software required for the RPC.
M.S.A, J.P.C.III, A.C., C.Foltz., A.I., L.N. developed the classical control system software. I.S.M. contributed to the design of transport waveform. 
S.F.C., R.D.D., J.M.D., B.E., C.Figgatt., J.P.G., A.H., A.A.H., C.J.K., N.K., M.M., A.R.M., A.J.P. and A.P.R developed and tuned the QPU, and maintained performance for the duration of the experiment.
A.R., J.G.B, and B.N. led the development of the QPU.
All authors contributed to technical discussions and the writing and editing of the manuscript and the appendices.

\section*{Data Availability}

The full data presented in this work is available at  \url{https://doi.org/10.5281/zenodo.17527281}.

\bibliography{refs}

\begin{thebibliography}{36}%
\makeatletter
\providecommand \@ifxundefined [1]{%
 \@ifx{#1\undefined}
}%
\providecommand \@ifnum [1]{%
 \ifnum #1\expandafter \@firstoftwo
 \else \expandafter \@secondoftwo
 \fi
}%
\providecommand \@ifx [1]{%
 \ifx #1\expandafter \@firstoftwo
 \else \expandafter \@secondoftwo
 \fi
}%
\providecommand \natexlab [1]{#1}%
\providecommand \enquote  [1]{``#1''}%
\providecommand \bibnamefont  [1]{#1}%
\providecommand \bibfnamefont [1]{#1}%
\providecommand \citenamefont [1]{#1}%
\providecommand \href@noop [0]{\@secondoftwo}%
\providecommand \href [0]{\begingroup \@sanitize@url \@href}%
\providecommand \@href[1]{\@@startlink{#1}\@@href}%
\providecommand \@@href[1]{\endgroup#1\@@endlink}%
\providecommand \@sanitize@url [0]{\catcode `\\12\catcode `\$12\catcode `\&12\catcode `\#12\catcode `\^12\catcode `\_12\catcode `\%12\relax}%
\providecommand \@@startlink[1]{}%
\providecommand \@@endlink[0]{}%
\providecommand \url  [0]{\begingroup\@sanitize@url \@url }%
\providecommand \@url [1]{\endgroup\@href {#1}{\urlprefix }}%
\providecommand \urlprefix  [0]{URL }%
\providecommand \Eprint [0]{\href }%
\providecommand \doibase [0]{https://doi.org/}%
\providecommand \selectlanguage [0]{\@gobble}%
\providecommand \bibinfo  [0]{\@secondoftwo}%
\providecommand \bibfield  [0]{\@secondoftwo}%
\providecommand \translation [1]{[#1]}%
\providecommand \BibitemOpen [0]{}%
\providecommand \bibitemStop [0]{}%
\providecommand \bibitemNoStop [0]{.\EOS\space}%
\providecommand \EOS [0]{\spacefactor3000\relax}%
\providecommand \BibitemShut  [1]{\csname bibitem#1\endcsname}%
\let\auto@bib@innerbib\@empty
\bibitem [{\citenamefont {Muthukrishnan}(2005)}]{Muthukrishnan2005}%
  \BibitemOpen
  \bibfield  {author} {\bibinfo {author} {\bibfnamefont {S.}~\bibnamefont {Muthukrishnan}},\ }\bibfield  {title} {\bibinfo {title} {Data streams: Algorithms and applications},\ }\href {https://doi.org/10.1561/0400000002} {\bibfield  {journal} {\bibinfo  {journal} {Foundations and Trends{\textregistered} in Theoretical Computer Science}\ }\textbf {\bibinfo {volume} {1}},\ \bibinfo {pages} {117–236} (\bibinfo {year} {2005})}\BibitemShut {NoStop}%
\bibitem [{\citenamefont {Kallaugher}\ \emph {et~al.}(2025)\citenamefont {Kallaugher}, \citenamefont {Parekh},\ and\ \citenamefont {Voronova}}]{kallaugher2024howto}%
  \BibitemOpen
  \bibfield  {author} {\bibinfo {author} {\bibfnamefont {J.}~\bibnamefont {Kallaugher}}, \bibinfo {author} {\bibfnamefont {O.}~\bibnamefont {Parekh}},\ and\ \bibinfo {author} {\bibfnamefont {N.}~\bibnamefont {Voronova}},\ }\bibfield  {title} {\bibinfo {title} {How to design a quantum streaming algorithm without knowing anything about quantum computing},\ }in\ \href {https://doi.org/10.1137/1.9781611978315.2} {\emph {\bibinfo {booktitle} {2025 Symposium on Simplicity in Algorithms (SOSA)}}}\ (\bibinfo {year} {2025})\ pp.\ \bibinfo {pages} {9--45}\BibitemShut {NoStop}%
\bibitem [{\citenamefont {Kallaugher}\ \emph {et~al.}(2023)\citenamefont {Kallaugher}, \citenamefont {Parekh},\ and\ \citenamefont {Voronova}}]{nadya}%
  \BibitemOpen
  \bibfield  {author} {\bibinfo {author} {\bibfnamefont {J.}~\bibnamefont {Kallaugher}}, \bibinfo {author} {\bibfnamefont {O.}~\bibnamefont {Parekh}},\ and\ \bibinfo {author} {\bibfnamefont {N.}~\bibnamefont {Voronova}},\ }\href@noop {} {\bibinfo {title} {Exponential quantum space advantage for approximating maximum directed cut in the streaming model}} (\bibinfo {year} {2023}),\ \Eprint {https://arxiv.org/abs/2311.14123} {arXiv:2311.14123 [quant-ph]} \BibitemShut {NoStop}%
\bibitem [{\citenamefont {Kallaugher}(2022)}]{kallaugher_2022_quantum_triangle_counting}%
  \BibitemOpen
  \bibfield  {author} {\bibinfo {author} {\bibfnamefont {J.}~\bibnamefont {Kallaugher}},\ }\bibfield  {title} {\bibinfo {title} {A quantum advantage for a natural streaming problem},\ }in\ \href {https://doi.org/10.1109/FOCS52979.2021.00091} {\emph {\bibinfo {booktitle} {2021 IEEE 62nd Annual Symposium on Foundations of Computer Science (FOCS)}}}\ (\bibinfo {year} {2022})\ pp.\ \bibinfo {pages} {897--908}\BibitemShut {NoStop}%
\bibitem [{\citenamefont {Gavinsky}\ \emph {et~al.}(2006)\citenamefont {Gavinsky}, \citenamefont {Kempe},\ and\ \citenamefont {de~Wolf}}]{BHM1}%
  \BibitemOpen
  \bibfield  {author} {\bibinfo {author} {\bibfnamefont {D.}~\bibnamefont {Gavinsky}}, \bibinfo {author} {\bibfnamefont {J.}~\bibnamefont {Kempe}},\ and\ \bibinfo {author} {\bibfnamefont {R.}~\bibnamefont {de~Wolf}},\ }\href@noop {} {\bibinfo {title} {Exponential separation of quantum and classical one-way communication complexity for a boolean function}} (\bibinfo {year} {2006}),\ \Eprint {https://arxiv.org/abs/quant-ph/0607174} {arXiv:quant-ph/0607174 [quant-ph]} \BibitemShut {NoStop}%
\bibitem [{\citenamefont {Kerenidis}\ and\ \citenamefont {Raz}(2006)}]{BHM2}%
  \BibitemOpen
  \bibfield  {author} {\bibinfo {author} {\bibfnamefont {I.}~\bibnamefont {Kerenidis}}\ and\ \bibinfo {author} {\bibfnamefont {R.}~\bibnamefont {Raz}},\ }\href@noop {} {\bibinfo {title} {The one-way communication complexity of the boolean hidden matching problem}} (\bibinfo {year} {2006}),\ \Eprint {https://arxiv.org/abs/quant-ph/0607173} {arXiv:quant-ph/0607173} \BibitemShut {NoStop}%
\bibitem [{\citenamefont {Gavinsky}\ \emph {et~al.}(2007)\citenamefont {Gavinsky}, \citenamefont {Kempe}, \citenamefont {Kerenidis}, \citenamefont {Raz},\ and\ \citenamefont {de~Wolf}}]{BHM3}%
  \BibitemOpen
  \bibfield  {author} {\bibinfo {author} {\bibfnamefont {D.}~\bibnamefont {Gavinsky}}, \bibinfo {author} {\bibfnamefont {J.}~\bibnamefont {Kempe}}, \bibinfo {author} {\bibfnamefont {I.}~\bibnamefont {Kerenidis}}, \bibinfo {author} {\bibfnamefont {R.}~\bibnamefont {Raz}},\ and\ \bibinfo {author} {\bibfnamefont {R.}~\bibnamefont {de~Wolf}},\ }\bibfield  {title} {\bibinfo {title} {Exponential separations for one-way quantum communication complexity, with applications to cryptography},\ }in\ \href {https://doi.org/10.1145/1250790.1250866} {\emph {\bibinfo {booktitle} {Proceedings of the Thirty-Ninth Annual ACM Symposium on Theory of Computing}}},\ \bibinfo {series and number} {STOC '07}\ (\bibinfo  {publisher} {Association for Computing Machinery},\ \bibinfo {address} {New York, NY, USA},\ \bibinfo {year} {2007})\ p.\ \bibinfo {pages} {516–525}\BibitemShut {NoStop}%
\bibitem [{\citenamefont {Montanaro}(2016)}]{10.5555/3179430.3179435}%
  \BibitemOpen
  \bibfield  {author} {\bibinfo {author} {\bibfnamefont {A.}~\bibnamefont {Montanaro}},\ }\bibfield  {title} {\bibinfo {title} {The quantum complexity of approximating the frequency moments},\ }\href@noop {} {\bibfield  {journal} {\bibinfo  {journal} {Quantum Info. Comput.}\ }\textbf {\bibinfo {volume} {16}},\ \bibinfo {pages} {1169–1190} (\bibinfo {year} {2016})}\BibitemShut {NoStop}%
\bibitem [{\citenamefont {Nayak}\ and\ \citenamefont {Touchette}(2017)}]{10.4230/lipics.ccc.2017.23}%
  \BibitemOpen
  \bibfield  {author} {\bibinfo {author} {\bibfnamefont {A.}~\bibnamefont {Nayak}}\ and\ \bibinfo {author} {\bibfnamefont {D.}~\bibnamefont {Touchette}},\ }\bibfield  {title} {\bibinfo {title} {Augmented index and quantum streaming algorithms for dyck(2)}\ }(\bibinfo  {publisher} {Schloss Dagstuhl – Leibniz-Zentrum f\"{u}r Informatik},\ \bibinfo {year} {2017})\BibitemShut {NoStop}%
\bibitem [{\citenamefont {Watrous}(1999)}]{watrous1999space}%
  \BibitemOpen
  \bibfield  {author} {\bibinfo {author} {\bibfnamefont {J.}~\bibnamefont {Watrous}},\ }\bibfield  {title} {\bibinfo {title} {Space-bounded quantum complexity},\ }\href {https://doi.org/https://doi.org/10.1006/jcss.1999.1655} {\bibfield  {journal} {\bibinfo  {journal} {Journal of Computer and System Sciences}\ }\textbf {\bibinfo {volume} {59}},\ \bibinfo {pages} {281} (\bibinfo {year} {1999})}\BibitemShut {NoStop}%
\bibitem [{goo(2025)}]{google2025quantum}%
  \BibitemOpen
  \bibfield  {title} {\bibinfo {title} {Quantum error correction below the surface code threshold},\ }\href {https://doi.org/https://doi.org/10.1038/s41586-024-08449-y} {\bibfield  {journal} {\bibinfo  {journal} {Nature}\ }\textbf {\bibinfo {volume} {638}},\ \bibinfo {pages} {920} (\bibinfo {year} {2025})}\BibitemShut {NoStop}%
\bibitem [{\citenamefont {Daguerre}\ \emph {et~al.}(2025)\citenamefont {Daguerre}, \citenamefont {Blume-Kohout}, \citenamefont {Brown}, \citenamefont {Hayes},\ and\ \citenamefont {Kim}}]{Daguerre2025}%
  \BibitemOpen
  \bibfield  {author} {\bibinfo {author} {\bibfnamefont {L.}~\bibnamefont {Daguerre}}, \bibinfo {author} {\bibfnamefont {R.}~\bibnamefont {Blume-Kohout}}, \bibinfo {author} {\bibfnamefont {N.~C.}\ \bibnamefont {Brown}}, \bibinfo {author} {\bibfnamefont {D.}~\bibnamefont {Hayes}},\ and\ \bibinfo {author} {\bibfnamefont {I.~H.}\ \bibnamefont {Kim}},\ }\bibfield  {title} {\bibinfo {title} {Experimental demonstration of high-fidelity logical magic states from code switching},\ }\bibfield  {journal} {\bibinfo  {journal} {Physical Review X}\ }\textbf {\bibinfo {volume} {15}},\ \href {https://doi.org/10.1103/dck4-x9c2} {10.1103/dck4-x9c2} (\bibinfo {year} {2025})\BibitemShut {NoStop}%
\bibitem [{\citenamefont {Córcoles}\ \emph {et~al.}(2021)\citenamefont {Córcoles}, \citenamefont {Takita}, \citenamefont {Inoue}, \citenamefont {Lekuch}, \citenamefont {Minev}, \citenamefont {Chow},\ and\ \citenamefont {Gambetta}}]{Crcoles2021}%
  \BibitemOpen
  \bibfield  {author} {\bibinfo {author} {\bibfnamefont {A.}~\bibnamefont {Córcoles}}, \bibinfo {author} {\bibfnamefont {M.}~\bibnamefont {Takita}}, \bibinfo {author} {\bibfnamefont {K.}~\bibnamefont {Inoue}}, \bibinfo {author} {\bibfnamefont {S.}~\bibnamefont {Lekuch}}, \bibinfo {author} {\bibfnamefont {Z.~K.}\ \bibnamefont {Minev}}, \bibinfo {author} {\bibfnamefont {J.~M.}\ \bibnamefont {Chow}},\ and\ \bibinfo {author} {\bibfnamefont {J.~M.}\ \bibnamefont {Gambetta}},\ }\bibfield  {title} {\bibinfo {title} {Exploiting dynamic quantum circuits in a quantum algorithm with superconducting qubits},\ }\bibfield  {journal} {\bibinfo  {journal} {Physical Review Letters}\ }\textbf {\bibinfo {volume} {127}},\ \href {https://doi.org/10.1103/physrevlett.127.100501} {10.1103/physrevlett.127.100501} (\bibinfo {year} {2021})\BibitemShut {NoStop}%
\bibitem [{\citenamefont {Foss-Feig}\ \emph {et~al.}(2023)\citenamefont {Foss-Feig}, \citenamefont {Tikku}, \citenamefont {Lu}, \citenamefont {Mayer}, \citenamefont {Iqbal}, \citenamefont {Gatterman}, \citenamefont {Gerber}, \citenamefont {Gilmore}, \citenamefont {Gresh}, \citenamefont {Hankin}, \citenamefont {Hewitt}, \citenamefont {Horst}, \citenamefont {Matheny}, \citenamefont {Mengle}, \citenamefont {Neyenhuis}, \citenamefont {Dreyer}, \citenamefont {Hayes}, \citenamefont {Hsieh},\ and\ \citenamefont {Kim}}]{2302.03029}%
  \BibitemOpen
  \bibfield  {author} {\bibinfo {author} {\bibfnamefont {M.}~\bibnamefont {Foss-Feig}}, \bibinfo {author} {\bibfnamefont {A.}~\bibnamefont {Tikku}}, \bibinfo {author} {\bibfnamefont {T.-C.}\ \bibnamefont {Lu}}, \bibinfo {author} {\bibfnamefont {K.}~\bibnamefont {Mayer}}, \bibinfo {author} {\bibfnamefont {M.}~\bibnamefont {Iqbal}}, \bibinfo {author} {\bibfnamefont {T.~M.}\ \bibnamefont {Gatterman}}, \bibinfo {author} {\bibfnamefont {J.~A.}\ \bibnamefont {Gerber}}, \bibinfo {author} {\bibfnamefont {K.}~\bibnamefont {Gilmore}}, \bibinfo {author} {\bibfnamefont {D.}~\bibnamefont {Gresh}}, \bibinfo {author} {\bibfnamefont {A.}~\bibnamefont {Hankin}}, \bibinfo {author} {\bibfnamefont {N.}~\bibnamefont {Hewitt}}, \bibinfo {author} {\bibfnamefont {C.~V.}\ \bibnamefont {Horst}}, \bibinfo {author} {\bibfnamefont {M.}~\bibnamefont {Matheny}}, \bibinfo {author} {\bibfnamefont {T.}~\bibnamefont {Mengle}}, \bibinfo {author} {\bibfnamefont {B.}~\bibnamefont {Neyenhuis}}, \bibinfo {author} {\bibfnamefont {H.}~\bibnamefont
  {Dreyer}}, \bibinfo {author} {\bibfnamefont {D.}~\bibnamefont {Hayes}}, \bibinfo {author} {\bibfnamefont {T.~H.}\ \bibnamefont {Hsieh}},\ and\ \bibinfo {author} {\bibfnamefont {I.~H.}\ \bibnamefont {Kim}},\ }\href@noop {} {\bibinfo {title} {Experimental demonstration of the advantage of adaptive quantum circuits}} (\bibinfo {year} {2023}),\ \Eprint {https://arxiv.org/abs/2302.03029} {arXiv:2302.03029 [quant-ph]} \BibitemShut {NoStop}%
\bibitem [{\citenamefont {B\"{a}umer}\ \emph {et~al.}(2024)\citenamefont {B\"{a}umer}, \citenamefont {Tripathi}, \citenamefont {Wang}, \citenamefont {Rall}, \citenamefont {Chen}, \citenamefont {Majumder}, \citenamefont {Seif},\ and\ \citenamefont {Minev}}]{Bumer2024}%
  \BibitemOpen
  \bibfield  {author} {\bibinfo {author} {\bibfnamefont {E.}~\bibnamefont {B\"{a}umer}}, \bibinfo {author} {\bibfnamefont {V.}~\bibnamefont {Tripathi}}, \bibinfo {author} {\bibfnamefont {D.~S.}\ \bibnamefont {Wang}}, \bibinfo {author} {\bibfnamefont {P.}~\bibnamefont {Rall}}, \bibinfo {author} {\bibfnamefont {E.~H.}\ \bibnamefont {Chen}}, \bibinfo {author} {\bibfnamefont {S.}~\bibnamefont {Majumder}}, \bibinfo {author} {\bibfnamefont {A.}~\bibnamefont {Seif}},\ and\ \bibinfo {author} {\bibfnamefont {Z.~K.}\ \bibnamefont {Minev}},\ }\bibfield  {title} {\bibinfo {title} {Efficient long-range entanglement using dynamic circuits},\ }\bibfield  {journal} {\bibinfo  {journal} {PRX Quantum}\ }\textbf {\bibinfo {volume} {5}},\ \href {https://doi.org/10.1103/prxquantum.5.030339} {10.1103/prxquantum.5.030339} (\bibinfo {year} {2024})\BibitemShut {NoStop}%
\bibitem [{\citenamefont {Carrera~Vazquez}\ \emph {et~al.}(2024)\citenamefont {Carrera~Vazquez}, \citenamefont {Tornow}, \citenamefont {Ristè}, \citenamefont {Woerner}, \citenamefont {Takita},\ and\ \citenamefont {Egger}}]{CarreraVazquez2024}%
  \BibitemOpen
  \bibfield  {author} {\bibinfo {author} {\bibfnamefont {A.}~\bibnamefont {Carrera~Vazquez}}, \bibinfo {author} {\bibfnamefont {C.}~\bibnamefont {Tornow}}, \bibinfo {author} {\bibfnamefont {D.}~\bibnamefont {Ristè}}, \bibinfo {author} {\bibfnamefont {S.}~\bibnamefont {Woerner}}, \bibinfo {author} {\bibfnamefont {M.}~\bibnamefont {Takita}},\ and\ \bibinfo {author} {\bibfnamefont {D.~J.}\ \bibnamefont {Egger}},\ }\bibfield  {title} {\bibinfo {title} {Combining quantum processors with real-time classical communication},\ }\href {https://doi.org/10.1038/s41586-024-08178-2} {\bibfield  {journal} {\bibinfo  {journal} {Nature}\ }\textbf {\bibinfo {volume} {636}},\ \bibinfo {pages} {75–79} (\bibinfo {year} {2024})}\BibitemShut {NoStop}%
\bibitem [{\citenamefont {Koch}\ \emph {et~al.}(2025)\citenamefont {Koch}, \citenamefont {Lawrence}, \citenamefont {Singhal}, \citenamefont {Sivarajah},\ and\ \citenamefont {Duncan}}]{koch2025guppy}%
  \BibitemOpen
  \bibfield  {author} {\bibinfo {author} {\bibfnamefont {M.}~\bibnamefont {Koch}}, \bibinfo {author} {\bibfnamefont {A.}~\bibnamefont {Lawrence}}, \bibinfo {author} {\bibfnamefont {K.}~\bibnamefont {Singhal}}, \bibinfo {author} {\bibfnamefont {S.}~\bibnamefont {Sivarajah}},\ and\ \bibinfo {author} {\bibfnamefont {R.}~\bibnamefont {Duncan}},\ }\href@noop {} {\bibinfo {title} {Guppy: pythonic quantum-classical programming}} (\bibinfo {year} {2025}),\ \Eprint {https://arxiv.org/abs/2510.12582} {arXiv:2510.12582 [quant-ph]} \BibitemShut {NoStop}%
\bibitem [{\citenamefont {Liu}\ \emph {et~al.}(2025)\citenamefont {Liu} \emph {et~al.}}]{jpmccertrand2}%
  \BibitemOpen
  \bibfield  {author} {\bibinfo {author} {\bibfnamefont {M.}~\bibnamefont {Liu}} \emph {et~al.},\ }\href@noop {} {\bibinfo {title} {Certified randomness amplification by dynamically probing remote random quantum states}},\ \bibinfo {howpublished} {To appear} (\bibinfo {year} {2025})\BibitemShut {NoStop}%
\bibitem [{\citenamefont {Kumar}(2019)}]{kumar2019practically}%
  \BibitemOpen
  \bibfield  {author} {\bibinfo {author} {\bibfnamefont {N.}~\bibnamefont {Kumar}},\ }\bibfield  {title} {\bibinfo {title} {Practically feasible robust quantum money with classical verification},\ }\bibfield  {journal} {\bibinfo  {journal} {Cryptography}\ }\textbf {\bibinfo {volume} {3}},\ \href {https://doi.org/10.3390/cryptography3040026} {10.3390/cryptography3040026} (\bibinfo {year} {2019})\BibitemShut {NoStop}%
\bibitem [{\citenamefont {Bar-Yossef}\ \emph {et~al.}(2004)\citenamefont {Bar-Yossef}, \citenamefont {Jayram},\ and\ \citenamefont {Kerenidis}}]{bar2004exponential}%
  \BibitemOpen
  \bibfield  {author} {\bibinfo {author} {\bibfnamefont {Z.}~\bibnamefont {Bar-Yossef}}, \bibinfo {author} {\bibfnamefont {T.~S.}\ \bibnamefont {Jayram}},\ and\ \bibinfo {author} {\bibfnamefont {I.}~\bibnamefont {Kerenidis}},\ }\bibfield  {title} {\bibinfo {title} {Exponential separation of quantum and classical one-way communication complexity},\ }in\ \href {https://doi.org/10.1145/1007352.1007379} {\emph {\bibinfo {booktitle} {Proceedings of the Thirty-Sixth Annual ACM Symposium on Theory of Computing}}},\ \bibinfo {series and number} {STOC '04}\ (\bibinfo  {publisher} {Association for Computing Machinery},\ \bibinfo {address} {New York, NY, USA},\ \bibinfo {year} {2004})\ p.\ \bibinfo {pages} {128–137}\BibitemShut {NoStop}%
\bibitem [{\citenamefont {Cross}\ \emph {et~al.}(2022)\citenamefont {Cross}, \citenamefont {Javadi-Abhari}, \citenamefont {Alexander}, \citenamefont {De~Beaudrap}, \citenamefont {Bishop}, \citenamefont {Heidel}, \citenamefont {Ryan}, \citenamefont {Sivarajah}, \citenamefont {Smolin}, \citenamefont {Gambetta},\ and\ \citenamefont {Johnson}}]{cross2022openqasm}%
  \BibitemOpen
  \bibfield  {author} {\bibinfo {author} {\bibfnamefont {A.}~\bibnamefont {Cross}}, \bibinfo {author} {\bibfnamefont {A.}~\bibnamefont {Javadi-Abhari}}, \bibinfo {author} {\bibfnamefont {T.}~\bibnamefont {Alexander}}, \bibinfo {author} {\bibfnamefont {N.}~\bibnamefont {De~Beaudrap}}, \bibinfo {author} {\bibfnamefont {L.~S.}\ \bibnamefont {Bishop}}, \bibinfo {author} {\bibfnamefont {S.}~\bibnamefont {Heidel}}, \bibinfo {author} {\bibfnamefont {C.~A.}\ \bibnamefont {Ryan}}, \bibinfo {author} {\bibfnamefont {P.}~\bibnamefont {Sivarajah}}, \bibinfo {author} {\bibfnamefont {J.}~\bibnamefont {Smolin}}, \bibinfo {author} {\bibfnamefont {J.~M.}\ \bibnamefont {Gambetta}},\ and\ \bibinfo {author} {\bibfnamefont {B.~R.}\ \bibnamefont {Johnson}},\ }\bibfield  {title} {\bibinfo {title} {Openqasm 3: A broader and deeper quantum assembly language},\ }\bibfield  {journal} {\bibinfo  {journal} {ACM Transactions on Quantum Computing}\ }\textbf {\bibinfo {volume} {3}},\ \href {https://doi.org/10.1145/3505636} {10.1145/3505636}
  (\bibinfo {year} {2022})\BibitemShut {NoStop}%
\bibitem [{\citenamefont {Ransford}\ \emph {et~al.}(2025)\citenamefont {Ransford} \emph {et~al.}}]{helios}%
  \BibitemOpen
  \bibfield  {author} {\bibinfo {author} {\bibfnamefont {A.}~\bibnamefont {Ransford}} \emph {et~al.},\ }\href@noop {} {\bibinfo {title} {Helios: A 98-qubit trapped-ion quantum computer}},\ \bibinfo {howpublished} {To appear} (\bibinfo {year} {2025})\BibitemShut {NoStop}%
\bibitem [{\citenamefont {Yoder}\ \emph {et~al.}(2025)\citenamefont {Yoder}, \citenamefont {Schoute}, \citenamefont {Rall}, \citenamefont {Pritchett}, \citenamefont {Gambetta}, \citenamefont {Cross}, \citenamefont {Carroll},\ and\ \citenamefont {Beverland}}]{yoder2025tour}%
  \BibitemOpen
  \bibfield  {author} {\bibinfo {author} {\bibfnamefont {T.~J.}\ \bibnamefont {Yoder}}, \bibinfo {author} {\bibfnamefont {E.}~\bibnamefont {Schoute}}, \bibinfo {author} {\bibfnamefont {P.}~\bibnamefont {Rall}}, \bibinfo {author} {\bibfnamefont {E.}~\bibnamefont {Pritchett}}, \bibinfo {author} {\bibfnamefont {J.~M.}\ \bibnamefont {Gambetta}}, \bibinfo {author} {\bibfnamefont {A.~W.}\ \bibnamefont {Cross}}, \bibinfo {author} {\bibfnamefont {M.}~\bibnamefont {Carroll}},\ and\ \bibinfo {author} {\bibfnamefont {M.~E.}\ \bibnamefont {Beverland}},\ }\href@noop {} {\bibinfo {title} {Tour de gross: A modular quantum computer based on bivariate bicycle codes}} (\bibinfo {year} {2025}),\ \Eprint {https://arxiv.org/abs/2506.03094} {arXiv:2506.03094 [quant-ph]} \BibitemShut {NoStop}%
\bibitem [{\citenamefont {Gidney}(2025)}]{gidney2025factor}%
  \BibitemOpen
  \bibfield  {author} {\bibinfo {author} {\bibfnamefont {C.}~\bibnamefont {Gidney}},\ }\href@noop {} {\bibinfo {title} {How to factor 2048 bit rsa integers with less than a million noisy qubits}} (\bibinfo {year} {2025}),\ \Eprint {https://arxiv.org/abs/2505.15917} {arXiv:2505.15917 [quant-ph]} \BibitemShut {NoStop}%
\bibitem [{\citenamefont {Litinski}(2019)}]{Litinski2019magicstate}%
  \BibitemOpen
  \bibfield  {author} {\bibinfo {author} {\bibfnamefont {D.}~\bibnamefont {Litinski}},\ }\bibfield  {title} {\bibinfo {title} {Magic {S}tate {D}istillation: {N}ot as {C}ostly as {Y}ou {T}hink},\ }\href {https://doi.org/10.22331/q-2019-12-02-205} {\bibfield  {journal} {\bibinfo  {journal} {{Quantum}}\ }\textbf {\bibinfo {volume} {3}},\ \bibinfo {pages} {205} (\bibinfo {year} {2019})}\BibitemShut {NoStop}%
\bibitem [{\citenamefont {Kretschmer}\ \emph {et~al.}(2025)\citenamefont {Kretschmer}, \citenamefont {Grewal}, \citenamefont {DeCross}, \citenamefont {Gerber}, \citenamefont {Gilmore}, \citenamefont {Gresh}, \citenamefont {Hunter-Jones}, \citenamefont {Mayer}, \citenamefont {Neyenhuis}, \citenamefont {Hayes} \emph {et~al.}}]{kretschmer2025demonstrating}%
  \BibitemOpen
  \bibfield  {author} {\bibinfo {author} {\bibfnamefont {W.}~\bibnamefont {Kretschmer}}, \bibinfo {author} {\bibfnamefont {S.}~\bibnamefont {Grewal}}, \bibinfo {author} {\bibfnamefont {M.}~\bibnamefont {DeCross}}, \bibinfo {author} {\bibfnamefont {J.~A.}\ \bibnamefont {Gerber}}, \bibinfo {author} {\bibfnamefont {K.}~\bibnamefont {Gilmore}}, \bibinfo {author} {\bibfnamefont {D.}~\bibnamefont {Gresh}}, \bibinfo {author} {\bibfnamefont {N.}~\bibnamefont {Hunter-Jones}}, \bibinfo {author} {\bibfnamefont {K.}~\bibnamefont {Mayer}}, \bibinfo {author} {\bibfnamefont {B.}~\bibnamefont {Neyenhuis}}, \bibinfo {author} {\bibfnamefont {D.}~\bibnamefont {Hayes}}, \emph {et~al.},\ }\href@noop {} {\bibinfo {title} {Demonstrating an unconditional separation between quantum and classical information resources}} (\bibinfo {year} {2025}),\ \Eprint {https://arxiv.org/abs/2509.07255} {arXiv:2509.07255 [quant-ph]} \BibitemShut {NoStop}%
\bibitem [{\citenamefont {Newman}(1991)}]{newman}%
  \BibitemOpen
  \bibfield  {author} {\bibinfo {author} {\bibfnamefont {I.}~\bibnamefont {Newman}},\ }\bibfield  {title} {\bibinfo {title} {Private vs. common random bits in communication complexity},\ }\href {https://doi.org/https://doi.org/10.1016/0020-0190(91)90157-D} {\bibfield  {journal} {\bibinfo  {journal} {Information Processing Letters}\ }\textbf {\bibinfo {volume} {39}},\ \bibinfo {pages} {67} (\bibinfo {year} {1991})}\BibitemShut {NoStop}%
\bibitem [{\citenamefont {Arora}\ and\ \citenamefont {Barak}(2009)}]{arora2009computational}%
  \BibitemOpen
  \bibfield  {author} {\bibinfo {author} {\bibfnamefont {S.}~\bibnamefont {Arora}}\ and\ \bibinfo {author} {\bibfnamefont {B.}~\bibnamefont {Barak}},\ }\href {https://doi.org/https://doi.org/10.1017/CBO9780511804090} {\emph {\bibinfo {title} {Computational complexity: a modern approach}}}\ (\bibinfo  {publisher} {Cambridge University Press},\ \bibinfo {year} {2009})\BibitemShut {NoStop}%
\bibitem [{\citenamefont {Bonami}(1970)}]{bonami1970etude}%
  \BibitemOpen
  \bibfield  {author} {\bibinfo {author} {\bibfnamefont {A.}~\bibnamefont {Bonami}},\ }\bibfield  {title} {\bibinfo {title} {{\'E}tude des coefficients de fourier des fonctions de $l^{p}(g)$},\ }in\ \href {https://doi.org/http://eudml.org/doc/74019} {\emph {\bibinfo {booktitle} {Annales de l'institut Fourier}}},\ Vol.~\bibinfo {volume} {20}\ (\bibinfo {year} {1970})\ pp.\ \bibinfo {pages} {335--402}\BibitemShut {NoStop}%
\bibitem [{\citenamefont {Beckner}(1975)}]{beckner1975inequalities}%
  \BibitemOpen
  \bibfield  {author} {\bibinfo {author} {\bibfnamefont {W.}~\bibnamefont {Beckner}},\ }\bibfield  {title} {\bibinfo {title} {Inequalities in fourier analysis},\ }\href {http://www.jstor.org/stable/1970980} {\bibfield  {journal} {\bibinfo  {journal} {Annals of Mathematics}\ }\textbf {\bibinfo {volume} {102}},\ \bibinfo {pages} {159} (\bibinfo {year} {1975})}\BibitemShut {NoStop}%
\bibitem [{\citenamefont {Kahn}\ \emph {et~al.}(1988)\citenamefont {Kahn}, \citenamefont {Kalai},\ and\ \citenamefont {Linial}}]{kahn1988influence}%
  \BibitemOpen
  \bibfield  {author} {\bibinfo {author} {\bibfnamefont {J.~D.}\ \bibnamefont {Kahn}}, \bibinfo {author} {\bibfnamefont {G.}~\bibnamefont {Kalai}},\ and\ \bibinfo {author} {\bibfnamefont {N.}~\bibnamefont {Linial}},\ }\bibfield  {title} {\bibinfo {title} {The influence of variables on boolean functions},\ }\href {https://api.semanticscholar.org/CorpusID:2347606} {\bibfield  {journal} {\bibinfo  {journal} {[Proceedings 1988] 29th Annual Symposium on Foundations of Computer Science}\ ,\ \bibinfo {pages} {68}} (\bibinfo {year} {1988})}\BibitemShut {NoStop}%
\bibitem [{\citenamefont {Shukla}\ and\ \citenamefont {Vedula}(2024)}]{shukla2023efficient}%
  \BibitemOpen
  \bibfield  {author} {\bibinfo {author} {\bibfnamefont {A.}~\bibnamefont {Shukla}}\ and\ \bibinfo {author} {\bibfnamefont {P.}~\bibnamefont {Vedula}},\ }\bibfield  {title} {\bibinfo {title} {An efficient quantum algorithm for preparation of uniform quantum superposition states},\ }\bibfield  {journal} {\bibinfo  {journal} {Quantum Information Processing}\ }\textbf {\bibinfo {volume} {23}},\ \href {https://doi.org/10.1007/s11128-024-04258-4} {10.1007/s11128-024-04258-4} (\bibinfo {year} {2024})\BibitemShut {NoStop}%
\bibitem [{\citenamefont {Herbert}\ \emph {et~al.}(2024)\citenamefont {Herbert}, \citenamefont {Sorci},\ and\ \citenamefont {Tang}}]{herbert2023almostoptimal}%
  \BibitemOpen
  \bibfield  {author} {\bibinfo {author} {\bibfnamefont {S.}~\bibnamefont {Herbert}}, \bibinfo {author} {\bibfnamefont {J.}~\bibnamefont {Sorci}},\ and\ \bibinfo {author} {\bibfnamefont {Y.}~\bibnamefont {Tang}},\ }\bibfield  {title} {\bibinfo {title} {Almost-optimal computational-basis-state transpositions},\ }\href {https://doi.org/10.1103/PhysRevA.110.012437} {\bibfield  {journal} {\bibinfo  {journal} {Phys. Rev. A}\ }\textbf {\bibinfo {volume} {110}},\ \bibinfo {pages} {012437} (\bibinfo {year} {2024})}\BibitemShut {NoStop}%
\bibitem [{\citenamefont {Maslov}(2016)}]{Maslov_2016}%
  \BibitemOpen
  \bibfield  {author} {\bibinfo {author} {\bibfnamefont {D.}~\bibnamefont {Maslov}},\ }\bibfield  {title} {\bibinfo {title} {Advantages of using relative-phase toffoli gates with an application to multiple control toffoli optimization},\ }\bibfield  {journal} {\bibinfo  {journal} {Physical Review A}\ }\textbf {\bibinfo {volume} {93}},\ \href {https://doi.org/10.1103/physreva.93.022311} {10.1103/physreva.93.022311} (\bibinfo {year} {2016})\BibitemShut {NoStop}%
\bibitem [{\citenamefont {Bravyi}\ \emph {et~al.}(2024)\citenamefont {Bravyi}, \citenamefont {Cross}, \citenamefont {Gambetta}, \citenamefont {Maslov}, \citenamefont {Rall},\ and\ \citenamefont {Yoder}}]{Bravyi2024}%
  \BibitemOpen
  \bibfield  {author} {\bibinfo {author} {\bibfnamefont {S.}~\bibnamefont {Bravyi}}, \bibinfo {author} {\bibfnamefont {A.~W.}\ \bibnamefont {Cross}}, \bibinfo {author} {\bibfnamefont {J.~M.}\ \bibnamefont {Gambetta}}, \bibinfo {author} {\bibfnamefont {D.}~\bibnamefont {Maslov}}, \bibinfo {author} {\bibfnamefont {P.}~\bibnamefont {Rall}},\ and\ \bibinfo {author} {\bibfnamefont {T.~J.}\ \bibnamefont {Yoder}},\ }\bibfield  {title} {\bibinfo {title} {High-threshold and low-overhead fault-tolerant quantum memory},\ }\href {https://doi.org/10.1038/s41586-024-07107-7} {\bibfield  {journal} {\bibinfo  {journal} {Nature}\ }\textbf {\bibinfo {volume} {627}},\ \bibinfo {pages} {778–782} (\bibinfo {year} {2024})}\BibitemShut {NoStop}%
\bibitem [{\citenamefont {Gidney}\ \emph {et~al.}(2024)\citenamefont {Gidney}, \citenamefont {Shutty},\ and\ \citenamefont {Jones}}]{gidney2024magic}%
  \BibitemOpen
  \bibfield  {author} {\bibinfo {author} {\bibfnamefont {C.}~\bibnamefont {Gidney}}, \bibinfo {author} {\bibfnamefont {N.}~\bibnamefont {Shutty}},\ and\ \bibinfo {author} {\bibfnamefont {C.}~\bibnamefont {Jones}},\ }\href@noop {} {\bibinfo {title} {Magic state cultivation: growing t states as cheap as cnot gates}} (\bibinfo {year} {2024}),\ \Eprint {https://arxiv.org/abs/2409.17595} {arXiv:2409.17595 [quant-ph]} \BibitemShut {NoStop}%
\end{thebibliography}%

\section*{Disclaimer}
This paper was prepared for informational purposes with contributions from the Global Technology Applied Research center of JPMorgan Chase \& Co. This paper is not a product of the Research Department of JPMorgan Chase \& Co. or its affiliates. Neither JPMorgan Chase \& Co. nor any of its affiliates makes any explicit or implied representation or warranty and none of them accept any liability in connection with this paper, including, without limitation, with respect to the completeness, accuracy, or reliability of the information contained herein and the potential legal, compliance, tax, or accounting effects thereof. This document is not intended as investment research or investment advice, or as a recommendation, offer, or solicitation for the purchase or sale of any security, financial instrument, financial product or service, or to be used in any way for evaluating the merits of participating in any transaction.

\appendix

\clearpage
\section{Classical Streaming Algorithm for Hidden Matching Problem}
\label{sec:classical_supp}
On an instance of $n$ vertices and matching parameter $\alpha$, \citet{BHM3} give a $\Theta(\sqrt{n/\alpha})$ lower bound for the one-way communication complexity of HM. Through a reduction of a communication problem to a streaming problem, this also gives an asymptotic lower bound for the space necessary to classically solve HM considered in this paper. Furthermore, in \cref{sec:classical_supp_lower_bound}, we give a stronger, exact version of this lower bound (Theorem~\ref{thm:phm-lower}), that determines the constant factors explicitly. %

We begin by first analyzing a simple algorithm and show that it has the same asymptotics as this bound. This asymptotically-optimal algorithm serves as a basis of comparison against the quantum sketch in \cref{fig:experiment}C and in \cref{tab:Resource_estimation_BHM_algorithm}. This analysis will also motivate the rigorous lower bounds below. Just as in the main text, we assume that the stream is ordered such that we receive vertex updates first and then edges. The algorithm fixes parameter $k \leq n$ representing the desired sketch size, and proceeds as follows:
\begin{enumerate}
    \item Randomly sample $k$ vertices and store them with their labels in the sketch $S$. 
    \item As we receive edge update $((u,v), w_{uv})$, check if $\{u,v\} \subset S$. If yes, then check $x_u \oplus x_v = w_{uv}$ and return the answer.
    \item If we process all edges and never find an edge contained in $S$, then return YES or NO with equal probability.
\end{enumerate}

The following probabilistic lemma forms the basis of this algorithm's analysis.
\begin{lemma}
    Let $M$ be an $\alpha n$ partial matching for $\alpha \in [0,1/4]$. For uniformly random $K \subseteq [n]$, the probability that $K$ does not contain any edge in $M$ is $\leq exp(-\alpha |K|^2/n)$. Consequently, for $\epsilon > 0$, choosing $|K| = \sqrt{\frac{-\ln(\epsilon)n}{\alpha}}$ ensures that this occurs with probability $\leq \epsilon$.
\end{lemma}

\begin{proof}
    For $j \in [|M|]$, let $E_j$ be the event that the $j$th edge $(u_j, v_j)$ does not have both vertices in $K$. Since $K$ is chosen uniformly at random, we have that the $\{E_j\}$ are independent. We are interested in $Pr[E_M]$ for $E_M = \cap_j E_j$. Suppose that $|M| = 1$. There are two cases: (i) $\{u,v\} \cap K = \emptyset$ and (ii) $|\{u,v\} \cap K| = 1$. The probability of the first case is given by $\binom{n-2}{|K|} / \binom{n}{|K|}$. In the second case, suppose $u \in K$ but not $v$. This happens with probability $\binom{n-2}{|K|-1} / \binom{n}{|K|}$. The full probability of (ii) is $2\binom{n-2}{|K|-1} / \binom{n}{|K|}$. So:

    \begin{align}
        Pr[E_1] = \frac{\binom{n-2}{|K|} + 2*\binom{n-2}{|K|-1}}{\binom{n}{|K|}} &= \frac{(n-|K|)(n+|K|-1)}{n(n-1)} \\
        &= 1-\frac{|K|(|K|-1)}{n(n-1)} \\
        &\approx 1 - \frac{|K|^2}{n^2}
    \end{align}

    Now suppose that we have $m = \alpha n \geq 1$ edges. By independence of $\{E_j\}$:

    \begin{align}
        Pr[ E_M] = \prod_{j=1}^{\alpha n}Pr[E_j] &= \left(1-\frac{|K|^2}{n^2} \right)^{\alpha n} \\
        &\approx \exp \left(-\alpha |K|^2/n \right)
    \end{align}

\noindent Choosing $|K|=\sqrt{\frac{-\ln(\epsilon)n}{\alpha}}$ results in $Pr[E_M] \leq \epsilon$.
\end{proof}

A simple application of this lemma gives the space lower bound for the classical algorithm. Moreover, this matches the lower bound given by \cite{BHM1,BHM2,BHM3}.

\begin{cor}
    By setting $k = \left\lceil \sqrt{\frac{\ln(3) n}{\alpha}} \right\rceil$ (with $\epsilon = 1/3$ above), the classical algorithm succeeds with probability 2/3.
\end{cor}

\noindent Note that the way the algorithm is presented, the total space is technically $k \lceil \log n \rceil$ (since we need $\lceil \log n \rceil$ bits to store each vertex in our sketch). Using Newman's Theorem \cite{newman} this may be reduced to $k + \lceil \log{n} \rceil$ bits -- however, the application of Newman's Theorem does not give a constructive algorithm. Furthermore, any style of algorithm that relies on detecting the existence of a matching pair relies on the probabilistic lemma, and so $k$ a lower-bound on the space complexity in general for this style of classical algorithm. We note that this matches up to constants the rigorous information-theoretic lower bounds to be discussed below.

\subsection{Classical Lower Bounds}
\label{sec:classical_supp_lower_bound}
Above, we determined the classical space necessary for the best-known classical algorithm from \cite{bar2004exponential}. In order to better compare against a quantum sketch, we also leverage techniques from communication complexity to provide an exact lower bound.

We define some notational preliminaries. We write $\{0,1\}^n$ for the Boolean cube, and identify $x\in\{0,1\}^n$ with a vector over $\mathbb{F}_2^n$.  
For a subset $S\subseteq[n]$, define the \emph{Walsh character} $\chi_S(x)=(-1)^{\sum_{i\in S} x_i}$, so that $\{\chi_S\}$ forms an orthonormal basis for functions $f:\{0,1\}^n\to\mathbb{R}$.  
The Fourier expansion of $f$ is
\begin{equation}
   f(x) = \sum_{S\subseteq[n]} \widehat f(S)\chi_S(x), \qquad \widehat f(S) = \mathbb{E}_x[f(x)\chi_S(x)]. 
\end{equation}
For a subset $A\subseteq\{0,1\}^n$, we define a Boolean indicator function $f=\mathbf{1}[x \in A]$ which tells whether an input $x$ belongs to the subset $A$. For such an indicator function $f$, the quantity 
\begin{equation}
 W^{(2)}(f) := \sum_{|S|=2}\widehat f(S)^2   
\end{equation}
is called the \emph{level-2 Fourier weight}, which measures the cumulative correlation of $f$ with all parities of pairs of bits.

\bigskip
In the following, we use $\alpha$-Partial Hidden Matching to refer to the Hidden Matching problem with implicit parameter $\alpha$.
\begin{theorem}[Classical Lower Bound for $\alpha$-Partial Hidden Matching]
\label{thm:phm-lower}
Let $n$ be even and $0<\alpha\le 1/4$.  
In the streaming model for the $\alpha$-Partial Hidden Matching problem, an algorithm must output an $\{i,j\}\in M$ together with the parity $x_i\oplus x_j$, given a stream of $x\in\{0,1\}^n$ and a matching $M$ of size $m=\alpha n$. Then any randomized algorithm with worst-case error at most $\epsilon<1/2$ must store 
\begin{equation}
   c\ge 
\frac{1}{e \ln 2}\Big(\tfrac12-\epsilon\Big)
\sqrt{\frac{n-1}{2\alpha}}
=\Omega\left(\sqrt{n/\alpha}\right) \text{ bits}. 
\end{equation}
\end{theorem}

\begin{proof}
By Yao's minimax principle~\cite{arora2009computational}, proving a lower bound for randomized protocols is equivalent to proving it for deterministic protocols under the hardest input distribution, which in this case is uniform over $\{0,1\}^n$.  
Thus we fix the public randomness, obtaining a deterministic encoder $E:\{0,1\}^n \to \{0,1\}^{\le c}$ where $c$ is the number of bits to store. The encoder partitions the cube into disjoint \emph{fibers}
\[
A_m = \{x : E(x)=m\},
\qquad \text{with}\quad
\sum_m |A_m| = 2^n.
\]

Since the algorithm can store $2^c$ distinct sketches, the pigeonhole principle implies that at least one fiber has density
\[
\gamma = \frac{|A_m|}{2^n} \ge 2^{-c}.
\]
We fix such an $m$ and denote $A:=A_m$.  

\medskip
Next, let $f=\mathbf{1}_A$ denote the indicator function of the fiber $A$.  
For each unordered pair $\{i,j\}$, define the parity character $\chi_{\{i,j\}}(x)=(-1)^{x_i\oplus x_j}$.
The Fourier coefficient $\widehat f(\{i,j\})=\mathbb{E}[f(x)\chi_{\{i,j\}}(x)]$ measures the correlation of the fiber with the parity $x_i\oplus x_j$. 

Since the prior on $x$ is uniform and $E$ is deterministic, Bob's posterior on $x$ given $m$ is uniform over $A$.
Hence the (posterior) bias of the parity on $\{i,j\}$ is
\begin{equation}
  b_{ij}
=\mathbb{E}[\chi_{\{i,j\}}(x)\mid x\in A]
=\frac{\mathbb{E}[\chi_{\{i,j\}}(x)f(x)]}{\mathbb{E}[f(x)]}
=\frac{\widehat f(\{i,j\})}{\gamma}.
\end{equation}

For a $\{\pm1\}$-valued random variable with mean $b$, the optimal constant predictor has success probability $(1+|b|)/2$; therefore the advantage over $1/2$ on the pair $\{i,j\}$ equals $|b_{ij}|/2$. If the matching $M$ has $m$ edges, the single-edge advantage is
\begin{equation}
\begin{split}
  \mathrm{Adv}(M\mid A)&=\max_{\{i,j\}\in M}\frac{|b_{ij}|}{2}
\le\frac{1}{2}\left(\sum_{\{i,j\}\in M} b_{ij}^2\right)^{1/2} \\
&=\frac{1}{2}\left(\sum_{\{i,j\}\in M}\frac{\widehat f(\{i,j\})^2}{\gamma^2}\right)^{1/2},  
\end{split}
\label{eq:adv}
\end{equation}
where we used $\max_k v_k \le \sqrt{\sum_k v_k^2}$ in the first inequality. Now, averaging over a uniformly random matching of size $m$ gives,
\begin{equation}
\begin{split}
  \mathbb{E}_M\Big[\sum_{\{i,j\}\in M}\widehat f(\{i,j\})^2\Big]
&=\frac{m}{\binom{n}{2}}\sum_{1\le i<j\le n}\widehat f(\{i,j\})^2 \\
&=\frac{m}{\binom{n}{2}}W^{(2)}(f),    
\end{split}
\end{equation}
and hence, by Jensen's inequality ($\mathbb{E}[\sqrt{X}]\le \sqrt{\mathbb{E}[X]}$) in \cref{eq:adv},
\begin{equation}
\mathbb{E}_M[\mathrm{Adv}(M\mid A)]\le
\frac{1}{2}\sqrt{\frac{m}{\binom{n}{2}}}\cdot \frac{\sqrt{W^{(2)}(f)}}{\gamma}.
\label{eq:advantage}
\end{equation}

Next, in order to control $W^{(2)}(f)$, we use the \emph{Bonami-Beckner hypercontractive inequality}~\cite{bonami1970etude, beckner1975inequalities}, 
as applied to Boolean functions by Kahn, Kalai, and Linial~\cite{kahn1988influence}. For a parameter $0 \le \rho \le 1$, we denote the \emph{noise operator} $T_\rho$ which acts on the indicator function $f$ by,
\begin{equation}
   (T_\rho f)(x)
=\mathbb{E}_y[f(y)], 
\end{equation}
where $y$ is obtained from $x$ by independently flipping each bit with probability $(1-\rho)/{2}$. In the Fourier domain, $T_\rho$ is diagonal,
\begin{equation}
 \widehat{T_\rho f}(S)=\rho^{|S|}\widehat f(S).   
\end{equation}

The Bonami--Beckner (hypercontractive) inequality states that for any Boolean function $f$ and any $0\le\rho\le1$,
\begin{equation}
   \|T_\rho f\|_2 \le \|f\|_{1+\rho^2}, 
\end{equation}
where $\|f\|_p=(\mathbb{E}[|f|^p])^{1/p}$.
Squaring both sides and expanding the left-hand side in the Fourier basis gives
\begin{equation}
    \sum_{S}\rho^{2|S|}\widehat f(S)^2
=\|T_\rho f\|_2^2
\le
\|f\|_{1+\rho^2}^2.
\end{equation}

For indicator functions $f=\mathbf{1}_A$ of density $\gamma=\mathbb{E}[f]=|A|/2^n$, we have that
$\|f\|_{1+\rho^2}^2 = \gamma^{2/(1+\rho^2)}$. 
Hence, for these functions,
\begin{equation}
\sum_{S}\rho^{2|S|}\widehat f(S)^2
\le
\gamma^{2/(1+\rho^2)}.  
\end{equation}
Grouping terms by their degree $k=|S|$ yields
\begin{equation}
\sum_{k=0}^{n}\rho^{2k}W^{(k)}(f)
\le
\alpha^{2/(1+\rho^2)},
\end{equation}
where $W^{(k)}(f):=\sum_{|S|=k}\widehat f(S)^2$.

All terms in this sum are nonnegative, so we can isolate the contribution from any single level.
For the level-$2$ Fourier weight we obtain
\begin{equation}
    \rho^{4}W^{(2)}(f) \le \gamma^{2/(1+\rho^2)} \Rightarrow W^{(2)}(f) \le
\rho^{-4}\gamma^{2/(1+\rho^2)}
\label{eq:walsh}
\end{equation}
Next, optimizing over $\rho$ yields the sharp bound 
\begin{equation}
W^{(2)}(f)\le e^2\gamma^2\ln^2\frac{1}{\gamma}.
\label{eq:level2}  
\end{equation}
To see this, let us denote $t = \rho^2 \in (0,1]$ and $l = \ln \frac{1}{\gamma}$. Then the expression in Eq~\ref{eq:walsh} can be rewritten as $\rho^{-4}\gamma^{2/(1+\rho^2)} = t^{-2}\exp\left(-\frac{2l}{1 + t}\right)$. Let us denote,
\begin{equation}
    g(t) = \ln\left(t^{-2}\exp\left(-\frac{2l}{1 + t}\right)\right) = -2\ln t - \frac{2l}{1+t}.
\end{equation}
We want to minimize $g(t)$ over $t\in (0,1]$. Computing $g'(t) = 0$ gives us two solutions,
\begin{equation}
    t = \frac{l-2 \pm l\sqrt{1 - 4/l}}{2}.
\end{equation}
In the large $l$ (or equivalently the small $\gamma$) regime, and given that $t \leq 1$ gives us the optimal choice of $t$ to be,
\begin{equation}
    t \approx \frac{1}{l} = \frac{1}{\ln 1/\gamma}.
\end{equation}
Plugging this in the above equation gives us the Eq~\ref{eq:level2},
\begin{equation}
    t^{-2}\exp\left(-\frac{2l}{1 + t}\right) \leq \ln^2 \frac{1}{\gamma}\gamma^2 e^2,
\end{equation}
where we have use the fact that $\exp(2l/(l+1)) \le e^2$.

Finally, substituting~\eqref{eq:level2} into~\eqref{eq:advantage} gives
\begin{equation}
   \mathbb{E}_M[\mathrm{Adv}(M\mid A)] \le (e \ln 2) c\sqrt{\frac{m} {\binom{n}{2}}}, 
   \label{eq:adv-bound}
\end{equation}
where we use the fact that $\gamma \ge 2^{-c}$, and thus $\ln(1/\gamma) \le c\ln2$.  If the overall error is at most $\epsilon$, then algorithm's success rate satisfies
\begin{equation}
    \tfrac12 - \epsilon \le  \mathbb{E}_M[\mathrm{Adv}(M\mid A)],
\end{equation}
which, together with \cref{eq:adv-bound} and $m = \alpha n$, yields
\begin{equation}
    c \ge
\frac{1}{e \ln 2}\Big(\tfrac12 - \epsilon\Big)
\sqrt{\frac{n-1}{2\alpha}}
=\Omega\big(\sqrt{n/\alpha}\big). 
\end{equation}
\end{proof}

\begin{cor}[Constant-Error Case]
For constant error $\epsilon=1/3$, any classical streaming algorithm for $\alpha$-Partiial Hidden Matching must satisfy
\begin{equation}
c\ge \frac{1}{(6e \sqrt{2})\ln 2} \sqrt{\frac{n-1}{\alpha}}. 
\end{equation}
\end{cor}

\section{Gate Complexity of Quantum Pair Sketch}
\label{sec:quantum_pair_sketch_complexity}

Ref.~\cite{kallaugher2024howto} unified all known quantum streaming algorithms with space advantage in the one-pass model using a quantum pair sketch. With this quantum sketch, streaming algorithms for hidden matching, triangle counting and maximum directed cut can be viewed as classical algorithms querying this sketch as a black box. In this Section, we summarize the quantum pair sketch and detail its logical gate complexities in the general case.

A quantum pair sketch $Q_T$ of size $k$ summarizes the set $T \subset \{0, 1\}^{k}$. The sketch has four primitive operations whose resource requirements are as follows:
\begin{enumerate}
    \item \texttt{create(T)}: takes as input the set $T$ and returns the sketch $Q_T$ which may be implemented as a uniform superposition of the basis states
    \begin{equation}
        Q_T = \frac{1}{\sqrt{|T|}}\sum_{t\in T}\ket{t}.
    \end{equation}
    The complexity of constructing such a state depends on the structure of $T$. In general, it takes $O(\log_2|T|)$ operations~\cite{shukla2023efficient}.  For HM, this may be constructed with a single layer of $k-1$ Hadamard gates. 
     \item \texttt{query\_one(a)}: probabilistically checks if $a \in T$. This can be achieved with a single ancilla as follows. First, apply a single C${}^{k+1}$X (multi-controlled Toffoli with $k$ controls) to implement the operation $\ket{x}\ket{0} \to \ket{x}\ket{\delta_{x=a}}$. Second, use a Toffoli gate followed by a measurement to project the state onto $\ket{a}\bra{a}$ or $1-\ket{a}\bra{a}$, thereby achieving the functionality of \texttt{query\_one}(a)

    \item \texttt{query\_pair(a,b)}: probabilistically checks if two states $a \neq b$ are both in $T$. This may be constructed using measurements with projectors $P_{\pm} \propto (\ket{a}\pm \ket{b})(\bra{a}\pm\bra{b}), P_0 = 1-P_+ - P_-$. This three-outcome measurement may be performed using two measurement ancilla. First, we apply a unitary operation $U_{ab}$ that performs the transformation $\ket{0} \to \ket{a}+\ket{b}$ and $\ket{x} \to \ket{a} - \ket{b}$ for some $x$. Then we use two $C^{k+1}X$ (multi-controlled Toffolis  with $k$ controls) on the two ancillas to project to $\ket{a} + \ket{b}$ and $\ket{a}-\ket{b}$ (by selecting for states $0^k$ and $x$ respectively). Finally, after measurement, we apply $U_{ab}^\dagger$. The unitary $U_{ab}$ can be implemented with one Hadamard and no more than $k$ CX gates. Overall, the \texttt{query\_pair(a,b)} operation takes two H gates,  $2k$ CX, two X gates (needed to select for $x$ for one of the multi-controlled Toffolis), and two $C^{k+1}X$ gates.
   
    \item \texttt{update($\pi$)}: takes the sketch $Q_T$ and returns a new sketch $Q_{T'}$ where the elements are permuted by $\pi$:
        \begin{equation}
        Q_{T'} = \frac{1}{\sqrt{|T|}}\sum_{t\in T}\ket{\pi(t)}.
    \end{equation}
    A general permutation can be composed into a series of transposition that transpose $a \leftrightarrow b$ leaving everything else unchanged. Each of those transposition may be implemented following the almost-optimal prescription in Ref.~\cite{herbert2023almostoptimal}. An example construction of the transposition is given in \cref{fig:update_unitary}. As suggested in the figure, this takes 2 H gates, $2k$ CX gates, $3k$ X gates, 2 C${}^{k+1}$X gates using the unitary construction from Ref.~\cite{herbert2023almostoptimal}.
\end{enumerate}

\begin{figure*}
\begin{align*}
    \begin{quantikz}
        \lstick{$\ket{x_1}$} & & \gate{X^{a_1 \oplus b_1}}\gategroup[4,steps=1,style={dashed,rounded corners,fill=yellow!20, inner xsep=2pt},background,label style={label position=above,anchor=north,yshift=+0.3cm}]{$U_{a,b}$} & \gate{X^{a_1 \oplus 1}}\gategroup[4,steps=3,style={dashed,rounded corners,fill=blue!20, inner xsep=2pt},background,label style={label position=above,anchor=north,yshift=+0.3cm}]{$\Pi_a$} & \ctrl{4} & \gate{X^{a_1 \oplus 1}}& \gate{X^{b_1 \oplus 1}}\gategroup[4,steps=3,style={dashed,rounded corners,fill=blue!20, inner xsep=2pt},background,label style={label position=above,anchor=north,yshift=+0.3cm}]{$\Pi_b$} & \ctrl{4} & \gate{X^{b_1 \oplus 1}} & \gate{X^{a_1 \oplus b_1}}\gategroup[4,steps=1,style={dashed,rounded corners,fill=yellow!20, inner xsep=2pt},background,label style={label position=above,anchor=north,yshift=+0.3cm}]{$U_{a,b}$} &&\\
        \lstick{$\ket{x_2}$} & & \gate{X^{a_2 \oplus b_2}} & \gate{X^{a_2 \oplus 1}} & \ctrl{3} & \gate{X^{a_2 \oplus 1}}& \gate{X^{b_2 \oplus 1}} & \ctrl{3} & \gate{X^{b_2 \oplus 1}} & \gate{X^{a_2 \oplus b_2}} &&\\
        \lstick{$\ket{x_3}$} & & \gate{X^{a_3 \oplus b_3}} & \gate{X^{a_3 \oplus 1}} & \ctrl{2} & \gate{X^{a_3 \oplus 1}}& \gate{X^{b_3 \oplus 1}} & \ctrl{2} & \gate{X^{b_3 \oplus 1}} & \gate{X^{a_3 \oplus b_3}} &&\\
        \lstick{$\ket{x_4}$} & & \gate{X^{a_4 \oplus b_4}} & \gate{X^{a_4 \oplus 1}} & \ctrl{1} & \gate{X^{a_4 \oplus 1}}& \gate{X^{b_4 \oplus 1}} & \ctrl{1} & \gate{X^{b_4 \oplus 1}} & \gate{X^{a_4 \oplus b_4}} &&\\
        \lstick{$\ket{0}$} & \gate{H} & \ctrl{-4} &  & \targ{} &  &   & \targ{} &  & \ctrl{-4} & \gate{H} &
    \end{quantikz}
\end{align*}
\caption{Realization of the transposition between $\ket{a} \leftrightarrow \ket{b}$ following \cite{herbert2023almostoptimal}. The $\Pi_{a/b}$ operations project the state to $\ket{a}\bra{a}$ and $\ket{b}\bra{b}$. Note that whenever $a_i = b_i = 0$, the $X$ gates in the inner layer on qubit $i$ cancel out, leading to slight reduction in the total gate count.}
\label{fig:update_unitary}
\end{figure*}
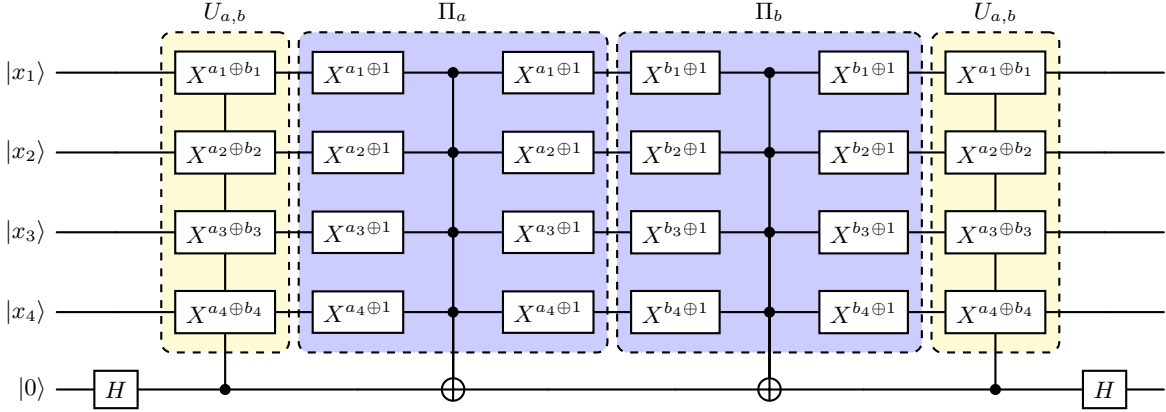

For specific problems with suitable structure, there may be more convenient methods of achieving the above operations. We provide HM-specific constructions with reduced cost in \cref{sec:hm-quantum-gate-counts}.

\section{Hybrid Quantum-Classical Streaming Algorithm for Hidden Matching Problem}
\label{sec:quantum-supp}
We now provide details of the quantum circuits used by the HM algorithm and the resources required to run it on quantum hardware. Due to the need to preserve a coherent quantum state throughout the process, we also analyze the resource requirements for running this algorithm fault-tolerantly in \cref{sec:FT_resource_estimation}.

\subsection{Space Requirements}
The quantum-classical hybrid sketching algorithm given in the main text requires $\lceil \log{n} \rceil + 2$ qubits to implement. Each quantum sketch gives a correct answer with probability $p_{\rm correct}=\alpha$, incorrect answer with probability $p_{\rm wrong}\leq \alpha/2$, and otherwise the null result $\perp$ (with probability $1-p_{\rm correct}-p_{\rm wrong} \in [1-\alpha-\alpha/2, 1-\alpha]$). We therefore need a majority voting procedure to ensure that the probability of getting an incorrect answer is suppressed (by boosting the $\geq \alpha/2$ separation between right and wrong answer). Let there be $k$ copies of the quantum sketch and let $p,q,k-p-q$ be the number of correct, incorrect, and $\perp$ outcomes, respectively. There are two ways for the vote to output the correct answer: (i) $p > q$ or (ii) if $p = q$ and the algorithm correctly guesses the right answer. This probability is given as
\begin{equation}\label{BHM_majority_vote_prob}
    \small 
    \begin{aligned}
        &\sum_{p=1}^k \sum_{q=1}^{\min{(k-p,p)}}  \frac{k!}{p! q! (k-p-q)!} \alpha^p \left(\frac{\alpha}{2}\right)^{q} \left(1-\frac{3\alpha}{2}\right)^{k-p-q} \\
        &+\frac{1}{2}\sum_{p=1}^{\lfloor k/2 \rfloor + 1}  \frac{k!}{p! q! (k-p-q)!} \frac{\alpha^{2p}}{2^{p}} \left(1-\frac{3\alpha}{2}\right)^{k-2p}.
    \end{aligned}
\end{equation}

We are interested in finding when \cref{BHM_majority_vote_prob} is greater than or equal to $2/3$. For $\alpha = 1/4$, this occurs at $k \geq 5$ (see \cref{fig:BHM_majority_vote}). Therefore, we need $5\cdot (\lceil \log(n)\rceil + 2)$ space to fully run the hybrid quantum-classical algorithm for HM with success probability $2/3$ noiselessly. Likewise, for general $\alpha$, the number of copies required is $k \leq 1.5/\alpha$.

\begin{figure}[!ht]
\includegraphics[width=\linewidth]{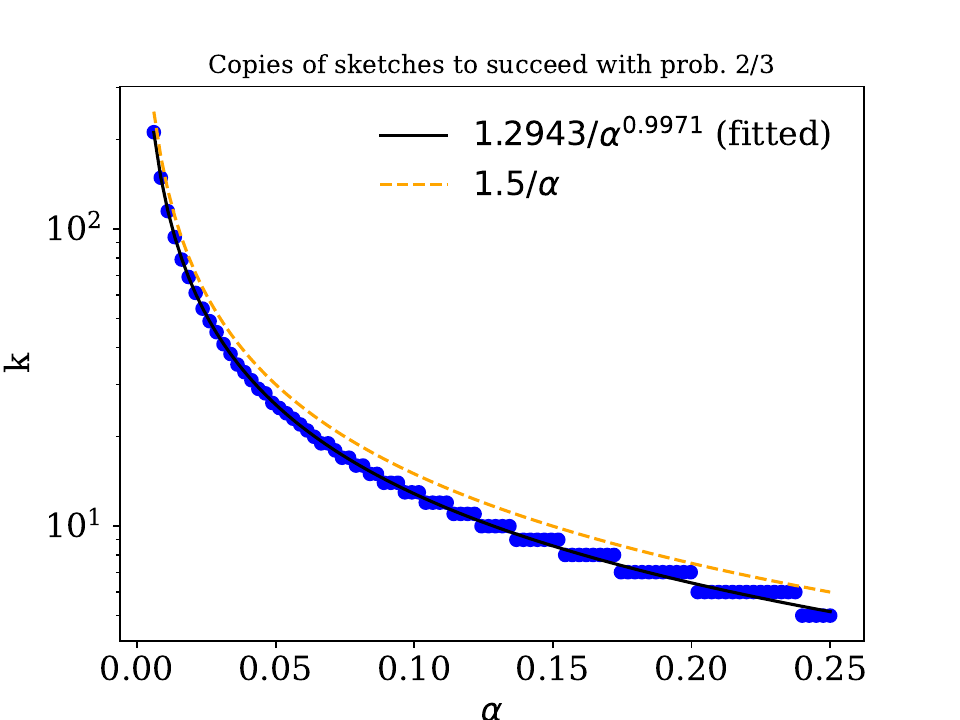}
    \caption{Smallest number of copies $k$ required for \cref{BHM_majority_vote_prob} to be $\geq 2/3$. Solid black line shows a fit that scales like $\Theta(1/\alpha)$. Dashed line in orange serves as a numeric upper-bound for $\alpha \in [0.01, 0.25]$.}
    \label{fig:BHM_majority_vote}
\end{figure}

\subsection{Gate Counts for HM Sketch}
\label{sec:hm-quantum-gate-counts}

We now discuss the operations necessary for a quantum sketch for HM in terms of the general operations enumerated in \cref{sec:quantum_pair_sketch_complexity}. The HM sketch is of size $\log(n) + 2$ bits, and is initialized (\texttt{create} in \cref{sec:quantum_pair_sketch_complexity}) as
\begin{equation}
    \frac{1}{\sqrt{2n}}\sum_{b\in \{0, 1\}} \sum_{x \in [n]}\ket{x, 0, b}
\end{equation}
this may be achieved  with H gates on the first $\log n$ qubits and the last qubit.

When the quantum sketch receives a vertex $(v, x_v)$, it performs an update $\ket{v, 0, b} \to \ket{v, x_v, b}$ for both $b \in \{0, 1\}$ . This corresponds to an  \texttt{update($\pi$)} operation with $\pi((v, 0, b))=(v, x_v, b)$ for $b \in \{0, 1\}$. While update operations are difficult in general, for HM, both of these transpositions (for the two values of $b$) can be achieved with a single $C^{\log(n) + 1}X$ gate, which conditionally flips the second-to-last qubit based on on the first $\log(n)$ qubits. 

Likewise, when the quantum sketch receives an edge update $(u, v, z_{uv})$, it performs four \texttt{query\_pair}(a,b) operations with
\begin{enumerate}
        \item $a=(u,0,0),b=(v,0,0)$
        \item $a=(u,0,1),b=(v,1,1)$
        \item $a=(u,1,1),b=(v,0,1)$
        \item $a=(u,1,0),b=(v,1,0)$
\end{enumerate}

We have $n$ vertex updates and each requires a single $C^{\log(n) + 1}X$ gate. Furthermore, each of the $\alpha n = n/4$ edge updates, requires four applications of \texttt{query\_pair}, each of which requires (as discussed in \cref{sec:quantum_pair_sketch_complexity}) 2 $H$, $\leq 2\log(n) + 2$ CNOTs, and two $C^{\log(n) + 3}X$ gates. Overall, the $n/4$ edge updates take a total of $2n$ $H$ gates, $2n(\log(n) + 2)$ CNOTs, and $2n$ many $C^{\log(n) + 3}X$ gates. However, note that, for worst-case circuits where even after the last edge, the quantum sketch fails to return a YES/NO answer, there is no need to uncompute with the basis transformation $U_{ab}^\dagger$. This reduces the Hadamard count by one and CNOT counts by $\log n + 2$.

Including the single \texttt{create} operation,  $n$ \texttt{update} operation and $4\times n/4$ \texttt{query\_pair} operation, the final worst-case, total gate counts are as follows:
\begin{itemize}
    \item  $(2n + \log n)$ H gates.
    \item  $(2n-1)(\log(n) + 2)$ CNOTs.
    \item  $n$ $C^{\log(n) + 1}X$ gates.
    \item  $2n$ $C^{\log(n) + 3}X$ gates.
\end{itemize}
Note that these are logical gate counts and apply for both non-fault-tolerant and fault-tolerant resource estimates. 

\subsection{Gate Counts for Non-Fault-Tolerant Implementation}

For the non-fault-tolerant implementation in our experiment, the multi-qubit Toffolis are decomposed into three-qubit Toffoli and relative Toffoli gates, following the prescription in \cite{Maslov_2016}. The three-qubit gates are further decomposed into single-qubit rotations and CNOT gates. The overall two-qubit gate count for different sizes of interest is presented in \cref{tab2:CX_counts}.

\begin{table}[!ht]
\begin{tabular}{|c|c|c|c|c|}
\hline
$n$  & T   & H     & CNOT              \\ \hline
4  & 212   &  98    & 196             \\ \hline
8  & 616   &  291    & 555             \\ \hline
16 & 1616  & 772   & 1434            \\ \hline
32 & 4000  & 1925  & 3513            \\ \hline
64 & 9536  &  4614   & 8312            \\ \hline
$n$  & $n(5+24\log(n))$ &  $(1+12n)\log n$ & \makecell{$20n\log(n)+8n$ \\ $-\log(n) - 2$} \\ \hline
\end{tabular}
\caption{Physical gate-counts for non-fault-tolerant HM sketch for worst-case circuits that don't terminate until the end of the data stream. Multi-controlled Toffolis are decomposed into single- and two-qubit gates following \cite{Maslov_2016}. The quantum program is written in terms of these gates.}
\label{tab2:CX_counts}
\end{table}

\section{Fault tolerant resource estimation for the HM algorithm}
\label{sec:FT_resource_estimation}

Here we detail the resource estimates for running the HM algorithm on a fault-tolerant quantum computer. 
We begin by estimating the target fidelity $\gamma$ required for quantum sketches to ensure that the quantum algorithm has a maximum tolerable failure probability.

As discussed in the main text, noisy quantum sketches with fidelity $\gamma$ raise the failure probability of the algorithm from $\delta(k, \alpha)$ to
\begin{equation}
    \delta'(k, \alpha, \gamma) = \delta(k, \alpha) + k \cdot (1-\gamma),
\end{equation}
where $k$ is the number of sketches used and $\alpha$ is the parameter of the HM problem.
To ensure the algorithm’s failure probability does not exceed a threshold (e.g., $1/3$), we numerically determine the maximum tolerable infidelity $1-\gamma$ for given values of $k$ and $\alpha$. Figure~\ref{fig:infidelity} illustrates this relationship for $\alpha = 1/4$.

\begin{figure}[tb]
    \centering
    \includegraphics[width=0.95\linewidth]{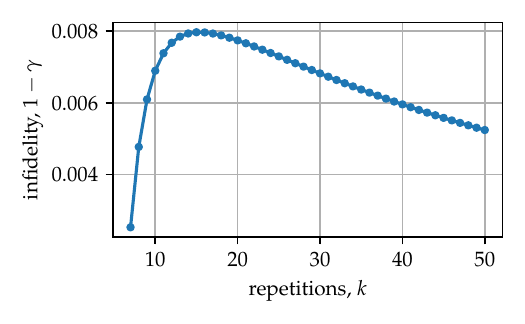}
    \caption{Numerical determination of tolerable infidelity $1-\gamma$ required to maintain failure probability of at most $1/3$ for $\alpha = 1/4$ and varying $k$, the number of repetitions of the sketch.}
    \label{fig:infidelity}
\end{figure}

Once the target infidelity $\gamma$ is established, we determine the required fidelity per gate, $F_*$, to ensure the overall algorithmic fidelity remains above $\gamma$. Suppose the total number of gates used in the implementation is $N_{\mathrm{gates}}$.
Then, the required gate fidelity is
\begin{align}
    F_{*} \geq 1 - \frac{1-\gamma}{N_{\mathrm{gates}}}.
\end{align}
For the HM algorithm, the total number of gates is given in Table~\ref{tab:gate_count}.

To achieve a target fidelity $F_*$ for the algorithm, both the memory error and the resource state (CCZ state) must have error rates below $1 - F_*$. The memory error is determined by error-correcting code parameters. For the surface code, this is achieved by setting the code distance to
\begin{equation}
    d = \left\lceil \frac{2\left(1 - \log(1 - F_*)\right)}{\log \left(p/p_{\mathrm{th}}\right)} \right\rceil,
\end{equation}
where $p$ is the physical error rate and $p_{\mathrm{th}} = 0.01$ is the threshold of the surface code.

For the two-gross code, we use code parameters and error rates provided in Table 2 of Ref.~\cite{yoder2025tour}.
Universal computation with the two-gross code is achieved through a modular assembly of code blocks, logical processing units (LPUs), and specialized adapters.
Each two-gross code module consists of 576 physical qubits for logical data storage, 158 qubits for the LPU, and 34 qubits for the code-code adapter.
The code-code adapter, which uses $2(d-1)$ qubits with code distance $d=18$, bridges the LPUs of adjacent modules via Bell-coupled ancilla qubits, enabling fault-tolerant measurement of joint logical Pauli operators between modules.
This inter-module connectivity is essential for scalable entangling operations and supports the execution of multi-qubit logical gates across the modular array.
The code-factory adapter, with a qubit count of $2d_\mathrm{factory}-1$ (where $d_\mathrm{factory}$ is the distance of the magic state factory), links code modules to dedicated magic state factories.
The architecture supports the use of magic states produced by surface code factories, which are well-benchmarked and compatible with the Bell-coupler infrastructure.

The total number of physical qubits for the HM algorithm is given by the sum of the total qubits required for memory and one resource state factory (assuming that resource states are consumed sequentially, one at a time).
For large problem sizes, the choice of error-correcting code and magic state factory protocol becomes critical. 
As shown in \cref{tab:Resource_estimation_BHM_algorithm_appendix}, for physical error rates of $p=1e^{-4}$, the logical error rates of the two-gross code (see Table 2 of Ref.~\cite{yoder2025tour}) satisfy the requirements of the HM algorithm for problem sizes $n \le 10^{13}$.
For $n > 10^{13}$, we employ the $\llbracket 360, 12, < 24 \rrbracket$ bivariate bicycle code~\cite{Bravyi2024}, assuming that its logical error rates remain below the threshold required by the HM algorithm.

Our resource estimates make the following assumptions about magic state factories:
\begin{itemize}
    \item For physical error rates of $10^{-4}$ and logical error rates above $10^{-14}$, we use magic state distillation protocols as described in~\cite{Litinski2019magicstate}.
    \item For physical error rates of $10^{-3}$ and logical error rates above $10^{-14}$, we use magic state cultivation techniques from~\cite{gidney2025factor}.
    \item Beyond these regimes, for $p = 10^{-4}$, we use the 8T-to-CCZ protocol from~\cite{gidney2025factor} and assume that T state infidelities of $10^{-10}$ to $10^{-12}$ are realistic, as supported by benchmarks in~\cite{Litinski2019magicstate}.
    \item For $p = 10^{-3}$, higher-fidelity T states can be produced by employing a color code of distance $d = 7$ (instead of $d = 5$), or by systematically growing to a $d = 7$ color code, which enables higher-fidelity injection stages. These T states can then be used to produce CCZ states of the required quality. However, these additional optimizations are expected to have only minor effects on the overall resource count and have not been included in our current estimates.
\end{itemize}

These assumptions provide a realistic and conservative basis for our resource estimates, but further improvements in code design and magic state preparation can further reduce the resource estimates for the quantum algorithm.

{
\renewcommand{\arraystretch}{1.5}
\begin{sidewaystable}
\centering
~\vspace{9cm}  %
\begin{longtable}{|c|c|c|c|c|c|c|c|c|c|c|c|} 
\hline
\multirow{3}{*}{$n$} 
& \multirow{3}{*}{\makecell{logical \\qubits}} 
& \multirow{3}{*}{\makecell{Toffoli\\ gates}} 
&\multirow{3}{*}{\makecell{CCZ \\ infidelity}} 
&\multicolumn{4}{c|}{surface code} 
& \multicolumn{1}{c|}{bivariate bicycle code} 
& \multicolumn{2}{c|}{classical bits} \\
\cline{5-11}
& & & & \multicolumn{2}{c|}{$p=10^{-3}$} & \multicolumn{2}{c|}{$p=10^{-4}$} 
& $p=10^{-4}$ 
&\multirow{2}{*}[0.5ex]{\makecell*[c]{best \\ known}} 
&\multirow{2}{*}[0.5ex]{\makecell*[c]{lower\\bound}}  \\
\cline{5-9}
& & & & distance & physical qubits & distance & physical qubits  & physical qubits & &
 \\
\hline
$10^4$    
& \numLogicalQubitsfourerrorratethree    
& $\gateCountfourerrorratefour$  
& $\errorRatebbfourerrorratefour$
&  \targetDistancefourerrorratethree
          &  $\numphysicalqubitsnewfourerrorratethree$ & \targetDistancefourerrorratefour 
          & $\numphysicalqubitsnewfourerrorratefour$    
          & $\numphysicalqubitsnewbbfourerrorratefour$ & $\classicalresourcefourerrorratethree$ 
          & $\classicalresourcebestfourerrorratethree$   \\
\hline
$10^5$    & \numLogicalQubitsfiveerrorratethree    
& $\gateCountfiveerrorratefour$    
& $\errorRatebbfiveerrorratefour$
&\targetDistancefiveerrorratethree
          &  $\numphysicalqubitsnewfiveerrorratethree$ & \targetDistancefiveerrorratefour 
          & $\numphysicalqubitsnewfiveerrorratefour$   
          & $\numphysicalqubitsnewbbfiveerrorratefour$ & $\classicalresourcefiveerrorratethree$ 
         & $\classicalresourcebestfourerrorratethree$ \\
\hline
$10^6$    
& \numLogicalQubitssixerrorratethree     
& $\gateCountsixerrorratefour$  
& $\errorRatebbsixerrorratefour$
&  \targetDistancesixerrorratethree
          &  $\numphysicalqubitsnewsixerrorratethree$ 
          & \targetDistancesixerrorratefour 
          & $\numphysicalqubitsnewsixerrorratefour$    
          & $\numphysicalqubitsnewbbsixerrorratefour$  
          & $\classicalresourcesixerrorratethree$ 
          & $\classicalresourcebestsixerrorratethree$ \\
\hline
$10^7$    
& \numLogicalQubitssevenerrorratethree   
& $\gateCountsevenerrorratefour$ 
& $\errorRatebbsevenerrorratefour$
&  \targetDistancesevenerrorratethree
          &  $\numphysicalqubitsnewsevenerrorratethree$ 
          & \targetDistancesevenerrorratefour 
          & $\numphysicalqubitsnewsevenerrorratefour$   
          & $\numphysicalqubitsnewbbsevenerrorratefour$ 
          & $\classicalresourcesevenerrorratethree$ 
          & $\classicalresourcebestsevenerrorratethree$ 
 \\
\hline
$10^8$    
& \numLogicalQubitseighterrorratethree   
& $\gateCounteighterrorratefour$  
& $\errorRatebbeighterrorratefour$
&  \targetDistanceeighterrorratethree
          &  $\numphysicalqubitsneweighterrorratethree$ & \targetDistanceeighterrorratefour 
          & $\numphysicalqubitsneweighterrorratefour$   
          & $\numphysicalqubitsnewbbeighterrorratefour$ 
          & $\classicalresourceeighterrorratethree$ 
          & $\classicalresourcebesteighterrorratethree$
\\
\hline
$10^9$    
& \numLogicalQubitsnineerrorratethree   
& $\gateCountnineerrorratefour$  
& $\errorRatebbnineerrorratefour$
&  \targetDistancenineerrorratethree
          &  $\numphysicalqubitsnewnineerrorratethree$ 
          & \targetDistancenineerrorratefour 
          & $\numphysicalqubitsnewnineerrorratefour$   
          & $\numphysicalqubitsnewbbnineerrorratefour$ 
          & $\classicalresourcenineerrorratethree$ 
          & $\classicalresourcebestnineerrorratethree$ 
 \\
\hline
$10^{10}$ 
& \numLogicalQubitstenerrorratethree    
& $\gateCounttenerrorratefour$    
& $\errorRatebbtenerrorratefour$
&  \targetDistancetenerrorratethree
          &  $\numphysicalqubitsnewtenerrorratethree$ 
          & \targetDistancetenerrorratefour 
          & $\numphysicalqubitsnewtenerrorratefour$    
          & $\numphysicalqubitsnewbbtenerrorratefour$ 
          & $\classicalresourcetenerrorratethree$ 
          & $\classicalresourcebesttenerrorratethree$ \\
\hline
$10^{11}$ 
& \numLogicalQubitselevenerrorratethree  
& $\gateCountelevenerrorratefour$ 
& $\errorRatebbelevenerrorratefour$
&  \targetDistanceelevenerrorratethree
          &  $\numphysicalqubitsnewelevenerrorratethree$ 
          & \targetDistanceelevenerrorratefour 
          & $\numphysicalqubitsnewelevenerrorratefour$   
          & $\numphysicalqubitsnewbbelevenerrorratefour$ 
          & $\classicalresourceelevenerrorratethree$ 
          & $\classicalresourcebestelevenerrorratethree$ \\
\hline
$10^{12}$ 
& \numLogicalQubitstwelveerrorratethree  
& $\gateCounttwelveerrorratefour$ 
& $\errorRatebbtwelveerrorratefour$
&  \targetDistancetwelveerrorratethree
          &  $\numphysicalqubitsnewtwelveerrorratethree$ 
          & \targetDistancetwelveerrorratefour 
          & $\numphysicalqubitsnewtwelveerrorratefour$   
          & $\numphysicalqubitsnewbbtwelveerrorratefour$ 
          & $\classicalresourcetwelveerrorratethree$ 
          & $\classicalresourcebesttwelveerrorratethree$ 
          \\
\hline
$10^{13}$ 
& \numLogicalQubitsthireteenerrorratethree 
& $\gateCountthireteenerrorratefour$ 
& $\errorRatebbthireteenerrorratefour$
&  \targetDistancethireteenerrorratethree
          &  $\numphysicalqubitsnewthireteenerrorratethree$ & \targetDistancethireteenerrorratefour 
          & $\numphysicalqubitsnewthireteenerrorratefour$   
          & $\numphysicalqubitsnewbbthireteenerrorratefour$ 
          & $\classicalresourcethireteenerrorratethree$ 
          & $\classicalresourcebestthireteenerrorratethree$ 
 \\
\hline
$10^{14}$ 
& \numLogicalQubitsfourteenerrorratethree 
& $\gateCountfourteenerrorratefour$  
& $\errorRatebbfourteenerrorratefour$
&  \targetDistancefourteenerrorratethree
          &  $\numphysicalqubitsnewfourteenerrorratethree$ & \targetDistancefourteenerrorratefour 
          & $\numphysicalqubitsnewfourteenerrorratefour$    
          & $\numphysicalqubitsnewbbfourteenerrorratefour$ & $\classicalresourcefourteenerrorratethree$ 
           & $\classicalresourcebestfourteenerrorratethree$ 
 \\
\hline
$10^{15}$ 
& \numLogicalQubitsfifteenerrorratethree  
& $\gateCountfifteenerrorratefour$ 
& $\errorRatebbfifteenerrorratefour$
&  \targetDistancefifteenerrorratethree
          &  $\numphysicalqubitsnewfifteenerrorratethree$ 
          & \targetDistancefifteenerrorratefour 
          & $\numphysicalqubitsnewfifteenerrorratefour$   
          & $\numphysicalqubitsnewbbfifteenerrorratefour$ 
          & $\classicalresourcefifteenerrorratethree$
        & $\classicalresourcebestfifteenerrorratethree$ 
 \\
\hline
\end{longtable}
\caption{\textbf{Fault-Tolerant Resource Estimates for the Hidden Matching (HM) Algorithm.}
Resource estimates for the HM quantum streaming algorithm using rotated surface codes and bivariate-bicycle codes, for a range of problem sizes $n$ and physical error rates ($p = 10^{-3}$ and $p = 10^{-4}$). 
The table lists logical qubits, Toffoli gates, code distances, and total physical qubit numbers required for both code families, as well as classical space requirements for the best-known and lower-bound classical algorithms.
For physical error rates of $p =10^{-4}$, the logical error rates of the two-gross code (see Table~2 of Ref.~\cite{yoder2025tour}) meet the requirements of the HM algorithm for problem sizes $n \leq 10^{13}$.
For $n > 10^{13}$, we employ the $\llbracket 360, 12, < 24 \rrbracket$ bivariate bicycle code from Ref.~\cite{Bravyi2024}, assuming an instruction set that achieves the logical error rates required for the HM algorithm.
Magic state factory fidelities are estimated using magic state distillation~\cite{Litinski2019magicstate} for $p = 10^{-4}$ and magic state cultivation~\cite{gidney2024magic} for $p = 10^{-3}$, with further plausible assumptions for T state infidelities and color code injection stages as described in the main text.}
\label{tab:Resource_estimation_BHM_algorithm_appendix}
\end{sidewaystable}
}

\end{document}